\documentclass[leqno, 12pt]{article}
\usepackage{times}
\usepackage{amssymb,amscd}
\usepackage{amsmath}
\usepackage{amsfonts}
\usepackage{mathrsfs}

\usepackage{amsthm}

\newtheorem{theorem}{Theorem}[section]

\newtheorem{definition}[theorem]{Definition}

\newtheorem{remark}[theorem]{Remark}

\newtheorem{lemma}[theorem]{Lemma}
\newtheorem{corollary}[theorem]{Corollary}

\begin{document}

\date{}

\title{Reasoning about proof and knowledge}

\author{Steffen Lewitzka
\thanks{Instituto de Matem\'atica e Estat\'istica,
Departamento de Ci\^encia da Computa\c c\~ao,
Universidade Federal da Bahia UFBA,
40170-110 Salvador -- BA,
Brazil,
e-mail: steffenlewitzka@web.de}}
\maketitle

\begin{abstract}
In previous work [Lewitzka, Log. J. IGPL 2017], we presented a hierarchy of classical modal systems, along with algebraic semantics, for the reasoning about intuitionistic truth, belief and knowledge. Deviating from G\"odel's interpretation of IPC in S4, our modal systems contain IPC in the way established in [Lewitzka, J. Log. Comp. 2015]. The modal operator can be viewed as a predicate for intuitionistic truth, i.e. proof. Epistemic principles are partially adopted from Intuitionistic Epistemic Logic IEL [Artemov and Protopopescu, Rev. Symb. Log. 2016]. In the present paper, we show that the S5-style systems of our hierarchy correspond to an extended Brouwer-Heyting-Kolmogorov interpretation and are complete w.r.t. a relational semantics based on intuitionistic general frames. In this sense, our S5-style logics are adequate and complete systems for the reasoning about proof combined with belief or knowledge. The proposed relational semantics is a uniform framework in which also IEL can be modeled. Verification-based intuitionistic knowledge formalized in IEL turns out to be a special case of the kind of knowledge described by our S5-style systems.
\end{abstract}



\section{Introduction}

In previous research \cite{lewjlc1, lewjlc2, lewsl, lewigpl}, we studied Lewis-style modal logics which have the property that strict equivalence $(\varphi\equiv\psi) := \square(\varphi\leftrightarrow\psi)$ satisfies the \textit{axioms of propositional identity}, i.e. certain identity axioms coming from Suszko's non-Fregean logics \cite{blosus, sus1}.\footnote{These natural axioms ensure that the identity connective $\equiv$ is a congruence relation modulo any given theory. We read $\varphi\equiv\psi$ as ``$\varphi$ and $\psi$ have the same meaning (denotation, \textit{Bedeutung})" or ``$\varphi$ and $\psi$ denote the same proposition". We strictly distinguish between formulas (syntactical objects) and propositions (semantic entities): a formula denotes a proposition. This is in accordance with our non-Fregean, intensional, view on logics: the semantics of a formula is, in general, more than its truth value. The Fregan Axiom $(\varphi\leftrightarrow\psi)\rightarrow (\varphi\equiv\psi)$ is not valid (see \cite{sus1}).} The condition `strict equivalence = propositional identity', satisfied by some logics in the vicinity of S1, particularly by S3--S5, warrants completeness w.r.t. a non-Fregean-style, algebraic semantics \cite{lewjlc1}. Logic $L$, introduced in \cite{lewjlc2}, has that property and combines classical propositional logic CPC with intuitionistic propositional logic IPC in the following sense:  
\begin{equation}\label{0}
\Phi\vdash_{IPC}\varphi \Leftrightarrow \square\Phi\vdash_{L}\square\varphi, \text{ where }\square\Phi:=\{\square\psi\mid\psi\in\Phi\},
\end{equation}
for any set of \textit{propositional} formulas $\Phi\cup\{\varphi\}$. In particular, $\varphi$ is a theorem of IPC iff $\square\varphi$ is a theorem of $L$. Thus, the map $\varphi\rightarrow\square\varphi$ is an embedding of IPC into $L$. $L$ is a classical modal logic for the reasoning about intuitionistic truth. The modal operator $\square$, applied to propositional formulas, can be regarded as a truth predicate for \textit{intuitionistic} truth in IPC. More precisely, for any prime theory $\Phi$ (i.e. rooted Kripke model) of IPC there is a model of $L$ such that for any propositional $\varphi$, $\square\varphi$ is classically true in the model of $L$ iff $\varphi$ is intuitionistically true in the corresponding Kripke model, i.e. $\varphi\in\Phi$. The other way round, every model of $L$ gives rise to such a corresponding prime theory of IPC.

More generally, the scheme $\square\varphi\leftrightarrow (\varphi\equiv\top)$ is valid for \textit{all formulas} $\varphi$ (see Theorem \ref{110} below), where $\varphi\equiv\top$ means that $\varphi$ holds intuitionistically. That scheme recalls the T-scheme (Convention T) of Tarski's truth theory. In our setting, the truth predicate $\square$ stands for intuitionistic truth and is an element of the object language.

Recall that the standard way to interpret IPC in a classical modal logic is given by G\"odel's translation $\varphi\mapsto\varphi'$ of IPC into S4, where $\varphi'$ results from $\varphi$ by prefixing every subformula with $\square$ (sometimes is used a different, though equivalent, translation). G\"odel showed that if $\varphi$ is a theorem of IPC, then $\varphi'$ is a theorem of S4. His conjecture that also the converse holds, i.e.
\begin{equation}\label{2}
\vdash_{IPC}\varphi \Leftrightarrow \text{ }\vdash_{S4}\varphi',
\end{equation}
was finally proved by McKinsey and Tarski (see, e.g., \cite{artbek, ottfei} for further historical background). The equivalence \eqref{2} particularly means that constructive reasoning is encoded in classical modal system S4 and can be recovered from it. In fact, G\"odel considered S4 as a provability calculus. Artemov \cite{art1} developed a logic of explicit proofs, called LP, giving an adequate semantics to G\"odel's provability calculus S4. LP represents an exact formalization of the BHK interpretation of IPC. Note, however, that S4 contains IPC only in codified form. That is, intuitionistic reasoning is not mirrored explicitly. In contrast, logic $L$ contains, in the sense of \eqref{0}, a faithful copy of IPC and reflects intuitionistic reasoning in a direct way. The modal operator is a predicate for instuitionistic truth: For any given model of $L$, `$\square\varphi$ holds classically $\Leftrightarrow$ $\varphi$ is intuitionistically true'. Of course, that equivalence does not hold in the context of S4. We believe these are properties that count in favor of $L$ as a calculus for the reasoning about proof as intuitionistic truth.

In this paper, we will argue that $L5$, i.e. the S5-style extension of $L$, is an adequate logic for the reasoning about proof in the sense of BHK semantics. This may appear surprising in view of the above mentioned classical results due to G\"odel, McKinsey, Tarski and Artemov which all rely on modal system S4 as a classical modal interpretation of intuitionistic logic. By results of \cite{lewjlc2}, IPC is contained in $L$ as a faithful copy via the embedding $\varphi\mapsto\square\varphi$. Intuitively, $L$ contains its own proof predicate: each formula $\square\varphi$ reads ``$\varphi$ is true in a constructive sense", even if $\square$ occurs in $\varphi$. In this paper, we shall describe that intuition by an \textit{extended} BHK interpretation with a proof-reading clause for the proof predicate itself, i.e. for formulas of the form $\square\varphi$. It turns out that the extended BHK semantics not only validates the axioms of $L$ but also the S5-style modal principles of the stronger logic $L5$. This means in particular that $L$ and its extensions $L3$ and $L4$ are not strong enough to formalize all reasoning principles valid under (extended) BHK semantics. A main result of our research is the construction of a relational semantics, based on intuitionistic general frames, which combines constructive reasoning with modal principles of system S5. This semantics can be viewed as a formal counterpart of the extended BHK interpretation. Our S5-style logics ($L5$ and its epistemic extensions) turn out to be sound and complete w.r.t. that relational semantics. This result formally confirms $L5$ as an adequate and complete classical logic for the reasoning about proof as intuitionistic truth. 

The existence of a proof predicate in the object language suggests an explicit distinction between \textit{actual proofs}, i.e. effected constructions, and \textit{possible proofs} as a kind of hypothetical constructions. $\square\varphi$ reads ``$\varphi$ has an actual proof", i.e. $\varphi$ is intuitionistically true. We say that ``$\varphi$ has a possible proof" (or ``a proof of $\varphi$ is possible"), notation: $\Diamond\varphi$, if $\neg\varphi$ has no actual proof, i.e. $\neg\square\neg\varphi$ holds classically. In this sense, the \textit{possibility of a proof} of a proposition means the absence of an actual proof of its falsehood. In a positive sense, we may understand a possible proof also as `conditions on a construction' instead of a construction itself.\footnote{We are inspired by discussions on the hypothetical judgment given in \cite{att1, att2}.} These conditions must not be in conflict with effected constructions. For example, consider $\square\neg\neg\varphi$ versus $\neg\square\neg\varphi$. The former formula says that $\neg\neg\varphi$ is intuitionistically true, intuitively: ``$\varphi$ cannot be false" or, in other words, ``$\neg\varphi$ has no possible proof". Of course, if $\neg\varphi$ has no possible proof, then $\neg\varphi$ has no actual proof. In classical logic $L$, we are able to formalize that intuitive fact by 
\begin{equation*}
\neg\Diamond\neg\varphi\rightarrow\neg\square\neg\varphi, \textit{ i.e.  }\square\neg\neg\varphi \rightarrow \neg\square\neg\varphi.
\end{equation*}
Actually, that formula is a theorem of $L$ (see the general case in item (iii) of Theorem \ref{130} below). Now, observe that that implication cannot be expressed in IPC, even if $\varphi$ is a pure propositional formula: while $\square\neg\neg\varphi$ corresponds to intuitionistic truth of $\neg\neg\varphi$, the weaker statement $\neg\square\neg\varphi$, ``$\neg\varphi$ is not intuitionistically true" =``$\varphi$ has a possible proof", corresponds to no formula in IPC. The more expressive logic $L$ makes the distinction between actual proof and possible proof explicit.\\

In \cite{lewigpl}, we enriched $L3$--$L5$ with epistemic axioms which are inspired by principles of Intuitionistic Epistemic Logic IEL introduced by Artemov and Protopopescu \cite{artpro}. IEL relies on the intuition that proof, as the strictest kind of a \textit{verification}, yields verification-based belief and knowledge. This is expressed by the axiom of intuitionistic co-reflection $\varphi\rightarrow K\varphi$. Furthermore, the classical knowledge axiom of reflection, $K\varphi\rightarrow\varphi$, is replaced by intuitionistic reflection $K\varphi\rightarrow\neg\neg\varphi$, which reads ``known propositions cannot be intuitionistically false".
IEL is in line with BHK semantics of IPC. G\"odel's translation $\varphi\mapsto\varphi'$ of IPC into classical modal logic extends to IEL, where, e.g., the axiom of intuitionistic co-reflection $\varphi\rightarrow K\varphi$ is translated as $\square(\square\varphi'\rightarrow\square K\square\varphi')$. This is in accordance with the proposed BHK clause for $K\varphi$ as ``a proof of $K\varphi$ is conclusive evidence of verification that $\varphi$ has a proof" \cite{artpro}, i.e. a proof of $K\varphi$ is a proof of a verification that $\varphi$ has a proof.\footnote{We would like to thank the anonymous referee for helpful comments.} According to that clause, $K\varphi$ can be read as ``it is verified that $\varphi$ has a proof". However, as pointed out in \cite{artpro}, IEL also captures the following possible reading: ``it is verified that $\varphi$ holds in some not specified constructive sense". We tend to the latter interpretation which seems to better harmonize with the informal and formal aspects of our approach. Accordingly, we will propose a weaker BHK clause for epistemic formulas $K\varphi$, one that is still compatible with basic principles of IEL, though is independent of the constraints imposed by G\"odel translation. 

Moreover, we propose in this paper a \textit{justification-based} interpretation of the kind of belief and knowledge modeled by our epistemic extensions of $L3$--$L5$ originally introduced in \cite{lewigpl}, in contrast to the \textit{verification-based} approach of IEL presented in \cite{artpro}. The agent believes/knows a proposition for some reason or justification. What the agent recognizes as a reason for her/his belief and knowledge depends on her/his internal conditions such as reasoning capabilities, experience, awareness etc. We assume the agent is rational enough to recognize an effected construction (actual proof) as a justification for belief and knowledge. However, the agent may be unable to associate the mere possibility of a proof of $\varphi$, i.e. the lack of an actual proof of $\neg\varphi$, with possible belief or knowledge of $\varphi$. That is, the mere possibility of a proof is not necessarily accepted as a justification. Under these assumptions, we have to reject intuitionistic co-reflection $\varphi\rightarrow K\varphi$ as a valid principle. Instead, its weaker classical version $\square\varphi\rightarrow \square K\varphi$, proposed in \cite{lewigpl}, is valid and ensures that actual proofs yield belief and knowledge. Nevertheless, our justification-based view on belief and knowledge is compatible with the verification-based one. If the agent accepts possible proofs as justifications for her/his possible belief and knowledge, then intuitionistic co-reflection $\varphi\rightarrow K\varphi$ is valid and all properties of IEL are restored. In the general case, however, a not perfectly rational agent may be unable or unwilling to recognize the mere possibility of a proof as a justification. 

In this sense, the verification-based approach turns out to be a special case of our justification-based approach. Actually, this is mirrored in our semantic formalization which is based on intuitionistic general frames: Both IEL and the epistemic extensions of $L5$ can be modeled and studied within the same framework of relational semantics. The exact relationship between both approaches now becomes explicit. In this framework of relational semantics, our modal version of co-reflection $\square\varphi\rightarrow\square K\varphi$ is valid. Intuitionistic co-reflection $\varphi\rightarrow K\varphi$, however, corresponds to a semantic condition that must be imposed as an additional constraint.

\section{An informal model of reasoning}

The aim of this section is to informally discuss a model that reflects our intuitions on (classical) reasoning about actual and possible proof, belief and knowledge.\footnote{We would like to point out that our modeling does not involve the concept of time or any dynamic behavior. A given model describes, from its classical point of view, a static situation regarding actual and possible proof, truth, belief and knowledge.} A formalization will be given in terms of the S5-style logics and their semantics presented in subsequent sections. Our purpose here is to illustrate that reasoning in these S5-style logics harmonizes with (a refined and extended version of) BHK interpretation. Recall that under  standard BHK interpretation,
\begin{itemize}
\item a proof of $\varphi\wedge\psi$ consists in a proof of $\varphi$ and a proof of $\psi$
\item a proof of $\varphi\vee\psi$ consists in a proof of $\varphi$ or a proof of $\psi$
\item a proof of $\varphi\rightarrow\psi$ consists in a construction that for any proof of $\varphi$ returns a proof of $\psi$
\item there is no proof of $\bot$.
\end{itemize}

Note that the proof-reading clause for implication involves implicitly the concept of `hypothetical proof' (see, e.g., \cite{att1, att2} for discussions). The establishment of $\varphi\rightarrow\psi$, in general, depends on `hypothetical constructions' for $\varphi$ and $\psi$, respectively. The distinction between actual and hypothetical proof becomes explicit in our more expressive, classical modal logics where $\square$ represents a proof predicate. In addition to the above example `$\square\neg\neg\varphi$ versus $\neg\square\neg\varphi$', we consider here the following two formulas: 
\begin{equation}\label{3}
\square(\varphi\rightarrow\psi)\text{ versus }\square\varphi\rightarrow\square\psi.
\end{equation}
While the former claims that a construction is established that converts any proof of $\varphi$ into a proof of $\psi$, the latter expresses a classical implication: the existence of an effected construction for $\varphi$ implies the existence of such a construction for $\psi$. Of course, the former statement is stronger than the latter. Hypothetical constructions are irrelevant in the second statement of \eqref{3}. That statement cannot be expressed by a formula in IPC.
We explicitly distinguish between the concepts of \textit{actual proof} as an effected construction (=intuitionistic truth) and \textit{possible proof} as a kind of hypothetical construction. If $\varphi$ has no actual proof, then the second statement of \eqref{3} is true. However, even if $\varphi$ has no actual proof, a proof of $\varphi$ may be possible and the first statement of \eqref{3} may be false ($\psi$ may have no possible proof). Recall that we say that $\varphi$ has a possible proof (or that a proof of $\varphi$ is possible) if there is no actual proof of $\neg\varphi$. Intuitively, a proposition has a possible proof if it is consistent with the actually proved propositions. In a constructive sense, we regard a possible proof as a hypothetical construction consistent with the effected constructions. More precisely, we consider a possible proof as a set of \textit{conditions on a construction} in a similar way as suggested by van Atten \cite{att1, att2}. Those conditions must not be in conflict with actual proofs. It will be helpful to view on possible proofs as concrete entities specified in that way.\footnote{As argued in \cite{att1, att2}, the notion of `hypothetical construction' can be avoided if one considers `conditions on a construction' instead of the construction itself: ``In order to establish $A\rightarrow B$, one has to conceive $A$ and $B$ as conditions on constructions, and to show that from the conditions specified by $A$ one obtains the conditions specified by $B$, according to transformations whose composition preserves mathematical constructibility." \cite{att1, att2}. Accordingly, we regard a \textit{possible proof} as a set of conditions that are not in conflict with effected constructions, i.e. with proved propositions. Of course, any actual proof also constitutes a possible proof.} 

Now, we are going to describe our informal model which, unsurprisingly, builds upon the intuition of Kripke semantics for IPC. We are given a set $\Phi_p$ of proved propositions, i.e. the propositions which have an actual proof (in this informal setting, we do not distinguish between formulas and propositions). Of course, all intuitionistic tautologies have actual proofs and are contained in $\Phi_p$. Moreover, since `actual proof = intuitionistic truth', we may think of $\Phi_p$ as a prime theory of IPC. $\Phi_p$ then represents the root of an intuitionistic Kripke model given by all prime theories that extend $\Phi_p$. According to our definition, a formula $\varphi$ has a possible proof iff $\neg\varphi$ has no actual proof iff $\neg\varphi\notin\Phi_p$ iff $\varphi$ belongs to some prime theory extending $\Phi_p$ iff $\varphi$ is contained in some maximal theory extending $\Phi_p$. 

Let us now further suppose that the propositional language contains a predicate $\square$ for actual proof. $\square\varphi$ reads ``there is an actual proof of $\varphi$", and $\Diamond\varphi:=\neg\square\neg\varphi$ reads ``there is a possible proof of $\varphi$". In order to model \textit{classical} reasoning about \textit{intuitionistic truth}, we choose a designated maximal theory $\Phi_m\supseteq\Phi_p$. $\Phi_m$ is the set of classically true propositions of our underlying model of reasoning. The predicate $\square$ for intuitionistic truth (=actual proof) should satisfy the following basic condition. For any $\varphi$,
\begin{equation}\label{5}
\square\varphi\in\Phi_m\Leftrightarrow \varphi\in\Phi_p.
\end{equation}
That is, `$\square\varphi$ is classically true' iff `$\varphi$ is intuitionistically true'. This biconditional is formalized in the object language by $\square\varphi\leftrightarrow (\varphi\equiv\top)$, see Theorem \ref{110} below.

We show that the following principles now are validated \textit{classically}, i.e. the respective formulas belong to $\Phi_m$:\\

(A1) $\square(\varphi\vee\psi)\rightarrow(\square\varphi\vee\square\psi)$. \eqref{5}, together with the fact that $\Phi_p$ is a prime theory, implies that all instances of (A1) belong to $\Phi_m$. Note that (A1) is in accordance with the BHK reading of disjunction. \\

(A2) $\square\varphi\rightarrow\varphi$. ``Actual proof implies classical truth." This follows from \eqref{5}. In fact, if $\square\varphi\in\Phi_m$, then $\varphi\in\Phi_p\subseteq \Phi_m$. Of course, (A2) is also plausible by our intuition on classical truth and actual proofs as effected constructions.\\

(K) $\square(\varphi\rightarrow\psi)\rightarrow(\square\varphi\rightarrow\square\psi)$. This is the principle of distribution for $\square$. Of course, $\varphi\rightarrow\psi\in\Phi_p$ together with $\varphi\in\Phi_p$ implies $\psi\in\Phi_p$. Hence, all instances of (K) belong to $\Phi_m$. \\

We expect that those principles are also \textit{intuitionistically} acceptable. In order to verify this by means of BHK semantics, we need an appropriate proof-reading clause for formulas of the form $\square\varphi$. We extend the standard BHK interpretation by the following clause for the proof predicate:
\begin{itemize}
\item A proof of $\square\varphi$ consists in presenting an actual proof of $\varphi$.
\end{itemize}

Since $\square\varphi$ reads ``there is an actual proof of $\varphi$", the above clause is plausible. We use the concept of `\textit{presenting} an actual proof' in an intuitive sense and leave its concrete meaning open. Nevertheless, we assume that it essentially relies on a \textit{proof-checking procedure} in the sense of Artemov \cite{art1} where proof-checking is established as a valid operation on proofs. We also assume that the procedure of checking a given \textit{actual} proof $s$ of $\varphi$ depends only on $s$ and $\varphi$. Then that procedure constitutes itself an actual proof, i.e. an effected construction. Consequently, by the above clause, a possible proof of $\square\varphi$ consists in an actual proof (see also Theorem \ref{130} (v) below). It follows that $\square\varphi$ has an actual proof or a proof of $\square\varphi$ is impossible. The latter means that $\neg\square\varphi$ has an actual proof. Thus: \\

\textit{Either there is an actual proof of $\square\varphi$ or there is an actual proof of $\neg\square\varphi$.}\\

This principle is also established by the following argumentation. Either there is an actual proof of $\varphi$ (an accessible object) or there is no actual proof of $\varphi$ (in other words: either $\varphi$ is intuitionistically true or $\varphi$ is not intuitionistically true). An actual proof of $\varphi$ yields (via proof-checking) an actual proof of $\square\varphi$. On the other hand, if $\varphi$ has no actual proof, then $\square\varphi$ cannot have any possible proof (in fact, a necessary condition on a construction for $\square\varphi$ is that $\varphi$ has an actual proof), and therefore $\neg\square\varphi = \square\varphi\rightarrow\bot$ has an actual proof: we may choose the identity function as an immediately given effected construction. Standard BHK semantics of disjunction then implies that $\square\varphi\vee\neg\square\varphi$ always has an actual proof, i.e. $\square (\square\varphi\vee\neg\square\varphi)$ is valid (see also Theorem \ref{130} (vii) below).\\

In the following, we list further examples of principles which are validated by extended BHK interpretation. Recall that we regard a proof of $\square\varphi$ as an effected construction that returns a proof-checked actual proof of $\varphi$. \\

(A4) $\square\varphi\rightarrow\square\square\varphi$.
Let $t$ be a proof of $\square\varphi$. As argued above, $t$ must be an actual proof. Then $t$ can be presented (proof-checked) by a procedure $u$. By definition, $u$ is a proof of $\square\square\varphi$. We have described a construction that converts any proof $t$ of $\square\varphi$ into a proof $u$ of $\square\square\varphi$.\\

(A5) $\neg\square\varphi\rightarrow\square\neg\square\varphi$. We describe a construction that for any proof $s$ of $\neg\square\varphi$ returns a proof $t$ of $\square\neg\square\varphi$. Let $s$ be a proof of $\neg\square\varphi$. Then $s$ is an actual proof or $s$ is a possible proof. If $s$ is only a possible proof, i.e. a set of conditions on a construction for $\neg\square\varphi =\square\varphi\rightarrow \bot$, then $s$ includes the condition that $\square\varphi$ has no proof. By definition, this condition must not be in conflict with actual proofs. By the discussion above, either $\square\varphi$ has an actual proof or $\neg\square\varphi$ has an actual proof. It follows that $\neg\square\varphi$ has an actual proof. In any case, $s$ implies the existence of an \textit{actual} proof of $\neg\square\varphi =\square\varphi\rightarrow \bot$. Then we have a concrete example of such a proof, namely the identity function as a trivial, immediately given, effected construction. Its presentation (proof-checking) is a procedure $t$ whose construction depends only on the given data. By definition, $t$ is a proof of $\square\neg\square\varphi$. Finally, a proof of $\neg\square\varphi\rightarrow\square\neg\square\varphi$ now is given by the function that for any (actual or possible) proof $s$ of $\neg\square\varphi$ returns the procedure $t$.\\

Also the principles (A1) and (A2) can be justified by extended BHK interpretation. This is clear for (A1) if one takes into account that an \textit{actual} proof of $\varphi\vee\psi$ requires, of course, \textit{actual} proofs of $\varphi$ or of $\psi$, respectively. Towards (A2), recall that a proof of $\square\varphi$ is a procedure that returns a proof-checked actual proof of $\varphi$. Then it is evident that there is a construction that for any proof of $\square\varphi$ returns a proof of $\varphi$. 

Finally, we not only justify distribution principle (K) but we show that the following stronger principle (A3) of modal logic S3 is intuitionistically acceptable in the sense of the extended BHK interpretation. We use here the following notation: if $s$ is a proof of $\square\chi$, then we also write $s.t$ instead of $s$ in order to express that $s$ is a construction that returns the proof-checked actual proof $t$ of $\chi$. \\

(A3) $\square(\varphi\rightarrow\psi)\rightarrow\square(\square\varphi\rightarrow\square\psi)$. Let $s.t$ be a proof of $\square(\varphi\rightarrow\psi)$, where $t$ is an actual proof of $\varphi\rightarrow\psi$. We show that $t$ gives rise to an actual proof of $\square\varphi\rightarrow\square\psi$. Let $u.v$ be a proof of $\square\varphi$, where $v$ is an actual proof of $\varphi$. Then $t(v)$ is an actual proof of $\psi$ and results in a proof $r.[t(v)]$ of $\square\psi$. We have described a function $q$, $s.t\mapsto q_{s.t}$, that for any proof $s.t$ of $\square(\varphi\rightarrow\psi)$ returns a function $q_{s.t}$. $q_{s.t}$ converts any proof $u.v$ of $\square\varphi$ into a proof $r.[t(v)]$ of $\square\psi$. Hence, $q_{s.t}$ is a proof of $\square\varphi\rightarrow\square\psi$. This shows that $\square(\varphi\rightarrow\psi)\rightarrow(\square\varphi\rightarrow\square\psi)$ is intuitionistically acceptable in the sense of extended BHK reading. All involved proofs are actual proofs. In particular, construction $q_{s.t}$ is an actual proof, and its presentation (proof-checking) yields a proof $p.[q_{s.t}]$ of $\square(\square\varphi\rightarrow\square\psi)$. This results in a construction that converts any proof $s.t$ of $\square(\varphi\rightarrow\psi)$ into a proof $p.[q_{s.t}]$ of $\square(\square\varphi\rightarrow\square\psi)$. Thus, (A3) is intuitionistically acceptable, too.\\

\textit{The modal axioms (A1)--(A5) are validated by the extended BHK interpretation. In this sense, they represent adequate principles for the reasoning about proof.} \\

However, (A4) and (A5) are not valid in our original model of reasoning. In order to fix this, we strengthen condition \eqref{5} to the following:
\begin{equation}\label{7}
\begin{split}
&\varphi\in\Phi_p\Leftrightarrow\square\varphi\in\Phi_p\\
&\varphi\not\in\Phi_p\Leftrightarrow\neg\square\varphi\in\Phi_p.
\end{split}
\end{equation}

Note that the biconditionals \eqref{5} remain valid. \eqref{7} implies $\square\varphi\vee\neg\square\varphi\in\Phi_p$, i.e. $\square\varphi\vee\neg\square\varphi$ is intuitionistically true, see Theorem \ref{130} (vii) below. Moreover, all instances of (A1)--(A5) are now intuitionistically true in the updated model. We show this only for the case of (A4). Let $\Psi\supseteq\Phi_p$ be a prime theory and suppose $\square\varphi\in\Psi$. Then by \eqref{7}, $\square\varphi\in\Phi_p$ and $\square\square\varphi\in\Phi_p\subseteq \Psi$. Thus, $\square\varphi\rightarrow\square\square\varphi\in\Phi_p$.\\

\textit{(A1)--(A5) represent adequate laws for the reasoning about proof, validated by extended BHK semantics and by our intuitive model of reasoning. Logic $L5$ contains the axioms (A1)--(A5) and can be seen as a formalization of that intuitive reasoning. $L5$ is (strongly) sound and complete w.r.t. a semantics of algebraic models that essentially correspond to our intuitive model of reasoning. Moreover, we will prove soundess and completeness of $L5$ w.r.t. a relational semantics based on intuitionistic general frames, i.e. a formal counterpart of informal BHK interpretation. In this sense, $L5$ is an adequate and complete modal logic for the classical reasoning about proof as intuitionistic truth.}\\

Intuitionistic Epistemic Logic (IEL) relies on the assumption that belief and knowledge are \textit{products of verification}, where verification is intuitively understood as ``evidence considered sufficiently conclusive for practical purposes" \cite{artpro}. Proof, as ``the most strict kind of verification", then yields \textit{verification-based} belief and knowledge. This is expressed by the axiom of intuitionistic co-reflection $\varphi\rightarrow K\varphi$. On the other hand, known propositions cannot be proved to be false. This is the principle of intuitionistic reflection, axiomatized by $K\varphi\rightarrow\neg\neg\varphi$. In the intuitionistic setting, those epistemic principles are in accordance with the BHK interpretation of IPC. Under the classical reading, however, they imply equivalence between knowledge and classical truth: $K\varphi\leftrightarrow\varphi$. The question arises in which way our classical modal logics $L3$--$L5$ can be extended by epistemic axioms such that intuitive principles of IEL are preserved or mirrored in some adequate way, and the metalogical implications \textit{intuitionistic truth} $\Rightarrow$ \textit{knowledge} $\Rightarrow$ \textit{classical truth} remain strict. We presented such epistemic extensions of $L3$--$L5$ in \cite{lewigpl} and proved completeness w.r.t. algebraic semantics. In a conceptual sense, however, it remained open which kind of belief and knowledge is described by those logics. We propose here a \textit{justification-based} interpretation of belief and knowledge as a generalization of the \textit{verification-based} approach to intuitionistic belief and knowledge given in \cite{artpro}. A proposition is believed or known for a given reason or justification.\footnote{Actually, this is the view on epistemic concepts that underlies certain Justification Logics, see \cite{artfit} for an overview.} What the agent accepts as a justification or reason is determined by her/his internal state. In any case, the agent should accept effected constructions as justifications, and justifications should be closed under Modus Ponens (otherwise, the agent would be too irrational). However, the agent may have very little confidence in possible proofs, and a proposition that has only a possible proof may appear unbelievable to the agent. Possible proofs are not necessarily recognized as justifications for (possible) belief and knowledge, intuitionistic co-reflection is not valid. However, if the agent recognizes possible proofs as justifications, then intuitionistic co-reflection is restored. In this sense, our justification-based view on belief and knowledge can be seen as a generalization of the verification-based approach of IEL presented in \cite{artpro}.

In the verification-based approach, intuitionistic belief and knowledge are understood as `products of verification'. Analogously, we consider here the following clause for a constructive reading of $K\varphi$:
\begin{itemize}
\item A proof of $K\varphi$ is the product of an epistemic justification of $\varphi$.
\end{itemize} 

Thus, $K\varphi$ can be read intuitionistically as ``$\varphi$ has an epistemic justification". We do not further specify here the intuitive concept of \textit{epistemic justification}. However, we assume that every actual proof constitutes an epistemic justification, and epistemic justifications are closed under Modus Ponens (this enables us to model justification sets as filters of a Heyting algebra).

Let us see how the proof-reading clause for $K\varphi$ together with the assumptions on epistemic justifications can be incorporated into our intuitive model of reasoning. The proof-reading clause is modeled by a correspondence between (possible or actual) proofs of $K\varphi$ and epistemic justifications of $\varphi$. For this purpose, we assume the existence of a function $E$ that assigns to each prime theory $\Psi\supseteq\Phi_p$ a set of propositions $E(\Psi)$ such that $\Phi_p\subseteq E(\Psi)$ and $E(\Psi)$ `is closed under Modus Ponens' (these conditions will ensure that $E(\Psi)$ corresponds to a filter, i.e. is a `theory of propositions'), and for all $\varphi$ holds:
\begin{equation*}
K\varphi\in\Psi\Leftrightarrow\varphi\in E(\Psi).
\end{equation*}
If $\varphi\in E(\Psi)$, then we say that $\varphi$ has an epistemic justification w.r.t. $\Psi$. We refer to $E(\Psi)$ as the justification set of $\Psi$.\\

The conditions imposed on function $E$ ensure that the extended intuitive model now validates the following epistemic principles:\\

(KBel) \textit{The distribution axiom $K(\varphi\rightarrow\psi)\rightarrow (K\varphi\rightarrow K\psi)$ holds intuitionistically}. In order to see this, it is enough to show that $K\varphi\rightarrow K\psi$ belongs to a prime theory $\Psi$ whenever $K(\varphi\rightarrow\psi)$ belongs to $\Psi$. This follows from the properties of $E$. The formalized statement is shown in the proof of Theorem \ref{940} below. \\

\textit{Every actual proof is recognized by the agent as an epistemic justification. Possible proofs, however, are in general too weak to be considered as justifications.}\\
In fact, by properties of function $E$, we have in particular $\Phi_p\subseteq E(\Phi_p)$. So if $\varphi$ has an actual proof, i.e. $\varphi\in\Phi_p$, then $\varphi\in E(\Phi_p)$ and $\varphi$ has a justification w.r.t. $\Phi_p$. On the other hand, $\Psi\subseteq E(\Psi)$ does not hold in general for arbitrary prime theories $\Psi\supseteq\Phi_p$.\\

(CoRe) \textit{The modal version of co-reflection, $\square\varphi\rightarrow\square K\varphi$, is intuitionistically validated by our model of reasoning.}\\
Let $\Psi\supseteq\Phi_p$ be a prime theory and suppose $\square\varphi\in\Psi$. Since (A4) is intuitionistically valid in our model, we get $\square\square\varphi\in\Psi$. By condition \eqref{7} above, $\square\varphi\in\Phi_p$ and $\varphi\in\Phi_p$. By properties of function $E$, $\Phi_p\subseteq E(\Phi_p)$. Thus, $\varphi\in E(\Phi_p)$ and therefore $K\varphi\in\Phi_p$ and $\square K\varphi\in\Phi_p\subseteq\Psi$. Thus, $\square\varphi\rightarrow\square K\varphi\in\Phi_p$.\\  

\textit{Intuitionistic reflection} (IntRe), $K\varphi\rightarrow\neg\neg\varphi$, \textit{is intuitionistically valid in our model iff the following holds: For any prime theory $\Psi$, the justification set $E(\Psi)$ is contained in every maximal theory that extends $\Psi$.}\\
This follows readily from properties of our model (see the proof of Theorem \ref{940} for a formalization). (IntRe) implies in particular that every justification set is consistent, i.e. contained in some maximal theory (`knowable propositions cannot be intuitionistically false'), and the justification set $E(\Phi_p)$ is contained in the maximal theory $\Phi_m$ (`intuitionistically known propositions are classically true').\\

\textit{We have shown that, additionally to (A1)--(A5), also the postulated epistemic principles are intuitionistically valid in our informal model of reasoning (where the validity of (IntRe) depends on additional semantic constraints). These are the principles of reasoning underlying our S5-style epistemic modal logics presented below.}\\

Finally, we would like to mention that (CoRe) $\square\varphi\rightarrow\square K\varphi$ is also validated by the extended BHK interpretation. Let $s.t$ be a proof of $\square\varphi$, where $t$ is an actual proof of $\varphi$. Then, as an actual proof, $t$ is recognized as an epistemic justification for $\varphi$. By the proof-reading clause for $K\varphi$, $t$ yields a proof $u$ of $K\varphi$. Since $t$ is an actual proof, we may assume that $u$ is an actual proof, too. Proof-checking yields a proof $v.u$ of $\square K\varphi$. We have described a construction that for any proof of $\square\varphi$ returns a proof of $\square K\varphi$.\\

The original axiom of intuitionistic co-reflection $\varphi\rightarrow K\varphi$ is the only principle from IEL that is not validated by our informal model. That axiom corresponds to the semantic condition 
\begin{equation}\label{8}
\Psi\subseteq E(\Psi),\text{ for all prime theories }\Psi\supseteq\Phi_p.
\end{equation}
\eqref{8} expresses that every (possible or actual) proof yields a justification. This is obviously stronger than
\begin{equation}\label{9}
\Psi_p\subseteq E(\Psi_p).
\end{equation}
Indeed, \eqref{9} only expresses that \textit{actual} proofs yield justifications. The comparison `\eqref{8} versus \eqref{9}' reveals semantically the difference between the verification-based approach of IEL and our justification-based approach. If we would impose the stronger condition \eqref{8} as a semantic constraint on our model of reasoning, then all epistemic principles of IEL would be intuitionistically true.

\section{Axiomatization and Algebraic Semantics}

The object language is inductively defined over an infinite set of variables $x_0$, $x_1$, ... , logical connectives $\wedge$, $\vee$, $\rightarrow$, $\bot$, the modal operator $\square$ and the epistemic operator $K$. $Fm$ is the set of all formulas, and $Fm_0\subseteq Fm$ is the set of all propositional formulas, i.e. those formulas of $Fm$ that contain neither the modal operator $\square$ nor the epistemic operator $K$. Finally, $Fm_1\subseteq Fm$ is the modal propositional language given by all formulas without epistemic operator $K$. We use the notation $\varphi[x:=\psi]$ to denote the formula that results from $\varphi$ by substituting simultaneously all occurrences of variable $x$ with formula $\psi$ (this is defined by induction on $\varphi$). Furthermore, we shall use the following abbreviations:\\

\noindent $\neg\varphi :=\varphi\rightarrow\bot$\\
$\top :=\neg\bot$\\
$\varphi\leftrightarrow\psi :=(\varphi\rightarrow\psi)\wedge (\psi\rightarrow\varphi)$\\
$\varphi\equiv\psi := \square(\varphi\leftrightarrow\psi)$ (``propositional identity = strict equivalence") \\
$\square\Phi :=\{\square\psi\mid\psi\in\Phi\}$, for $\Phi\subseteq Fm$\\
$\Diamond\varphi :=\neg\square\neg\varphi$\\

We consider the following list of \textbf{Axiom Schemes}\\

\noindent (INT) all theorems of IPC and their substitution-instances\footnote{A substitution-instance of $\varphi$ is the result of uniformly replacing variables in $\varphi$ by formulas of $Fm$.}\\
(A1) $\square(\varphi\vee\psi)\rightarrow(\square\varphi\vee\square\psi)$ (disjunction property) \\
(A2) $\square\varphi\rightarrow\varphi$\\
(A3) $\square(\varphi\rightarrow\psi)\rightarrow\square(\square\varphi\rightarrow\square\psi)$\\
(A4) $\square\varphi\rightarrow\square\square\varphi$\\
(A5) $\neg\square\varphi\rightarrow\square\neg\square\varphi$\\
(KBel) $K(\varphi\rightarrow\psi)\rightarrow (K\varphi\rightarrow K\psi)$ (distribution of belief)\\
(CoRe) $\square\varphi\rightarrow \square K\varphi$ (co-reflection)\\
(IntRe) $K\varphi\rightarrow\neg\neg\varphi$ (intuitionistic reflection)\\
(E4) $K\varphi\rightarrow K K\varphi$ (positive introspection) \\
(E5) $\neg K\varphi\rightarrow K \neg K\varphi$ (negative introspection) \\ 
(PNB) $K\varphi\rightarrow\square K\varphi$ (positive necessitation of belief) \\
(NNB) $\neg K\varphi\rightarrow\square\neg K\varphi$ (negative necessitation of belief)\\

\noindent and the following \textbf{Theorem Scheme} (TND) of \textit{tertium non datur}\\

\noindent (TND) $\varphi\vee\neg\varphi$.\\

Our inference rules are Modus Ponens MP ``From $\varphi$ and $\varphi\rightarrow\psi$ infer $\psi$", and Axiom Necessitation AN ``If $\varphi$ is an axiom, then infer $\square\varphi$". Note that rule AN applies only to the axioms of a given system but not to the theorems; in particular, AN does not apply to the theorems of the form (TND).

On this basis, we define a hierarchy of deductive systems. We shall see that in the context of our S5-style logics, i.e. those systems of our hierarchy containing (A4) and (A5), the axiomatization can be slightly simplified, see Corollary \ref{90} and Theorem \ref{140} (vii). For instance, (A3) can be replaced with the usual distribution law, and (CoRe) can be replaced with the simpler scheme $\square\varphi\rightarrow K\varphi$. The epistemic axioms (E4), (E5), (PNB), (NNB) seem to have no plausible validation under the extended BHK interpretation (epistemic justifications are, in general, not strong enough to warrant (E4) and (E5)). We regard those axioms as additional epistemic laws that go beyond the established BHK validated principles. Moreover, we shall see that modulo the modal axioms of S5, (NNB) derives from (PNB), see Theorem \ref{140} (vii).

Our deductive systems are based on principles of Lewis modal logics. Recall that the systems S1--S3 were originally proposed by C. I. Lewis as formalizations of the concept of \textit{strict implication} $\square(\varphi\rightarrow\psi)$ (see, e.g., \cite{hugcre} for a discussion). In the language of propositional modal logic $Fm_1$, Lewis system S1 can be axiomatized in the following way (on a basis due to E. J. Lemmon). The axioms are given by all classical tautologies along with their substitution-instances + axiom (A2) + the following transitivity axiom
\begin{equation}\label{20}
\square(\varphi\rightarrow\psi)\rightarrow (\square(\psi\rightarrow\chi)\rightarrow\square(\varphi\rightarrow\chi)).
\end{equation}
As inference rules we have MP, AN and the rule of Substitution of Proved Strict Equivalents SPSE ``From $\varphi\equiv\psi$ infer $\chi[x:=\varphi]\equiv\chi[x:=\psi]$".\\   

In \cite{lewjlc1, lewjlc2, lewsl}, we proposed to interpret \textit{strict equivalence} $\square(\varphi\leftrightarrow\psi)$ as \textit{propositional identity} $\varphi\equiv\psi$. That is, we considered Lewis-style modal logics where strict equivalence $(\varphi\equiv\psi):=\square(\varphi\leftrightarrow\psi)$ satisfies the following identity axioms: \\

\noindent (Id1) $(\varphi\equiv\varphi)$\\
(Id2) $(\varphi\equiv\psi)\rightarrow (\varphi\leftrightarrow\psi)$\\
(Id3) $(\varphi\equiv\psi)\rightarrow (\chi[x:=\varphi]\equiv\chi[x:=\psi])$\\

These axioms come from Suszko's \textit{non-Fregean logics} where the identity connective $\equiv$ is given as a primitive symbol of the language (see, e.g, \cite{blosus}). If the symbol $\equiv$ is defined as strict equivalence, then (Id1) and (Id2) are obviously satisfied in system S1. We refer to axiom (Id3) as the Substitution Principle SP. Obviously, SP is stronger than inference rule SPSE of S1. We showed in \cite{lewjlc1} that S1+SP, i.e. the system that results from S1 by adding all formulas of the form SP as theorems, has a natural algebraic semantics that extends straightforwardly to semantics for S3--S5. Recall that Lewis system S3 results from S1 by adding (A3) as axiom scheme (see \cite{hugcre}). It turns out that S3 is the weakest Lewis modal system where strict equivalence not only satisfies (Id1) and (Id2) but also (Id3) (see \cite{lewjlc1, lewsl}). Thus, we have `strict equivalence = propositional identity' in S3, a condition that also holds in the modal systems studied in the present paper. As a consequence, we may adopt the style of non-Fregean semantics presented in \cite{lewjlc1}. Now, let us define the modal logics that are relevant for the present research.

\begin{itemize}
\item The weakest system in our hierarchy is logic $L$, formalized in the propositional modal language $Fm_1$. $L$ is given by the axiom schemes (INT), (A1), (A2) and \eqref{20}. All formulas of the form (TND) and SP (i.e. (Id3)) are added as theorems. The inference rules are MP and AN. (Recall that AN applies only to axioms of the given system.)
\item Modal system $L3$ is given by the axiom schemes (INT), (A1), (A2), (A3), theorem scheme (TND) and the rules of MP and AN. $L4$ results from $L3$ by adding (A4) as axiom scheme, and $L5$ results from $L4$ by adding (A5) as axiom scheme.
\item Now, we consider the full language $Fm$ and define epistemic and modal extensions of $L3$ by adding further axiom schemes:
\begin{itemize}
\item $EL3^-=L3+(KBel)+(CoRe)$
\item $EL4^-=EL3^-+(A4) = L4+(KBel)+(CoRe)$
\item $EL5^-=EL4^-+(A5) = L5+(KBel)+(CoRe)$
\item $E4Ln^-=ELn^-+(E4)$, $n\in\{3,4,5\}$
\item $E5Ln^-=E4Ln^-+(E5)$, $n\in\{3,4,5\}$
\item $E6Ln^-=ELn^-+(PNB)+(NNB)$, $n\in\{3,4,5\}$
\item $ELn=ELn^-+(IntRe)$, $n\in\{3,4,5\}$
\item $EkLn=EkLn^-+(IntRe)$, $k\in\{4,5,6\}$ and $n\in\{3,4,5\}$.
\end{itemize} 
Systems containing scheme (IntRe) are regarded as logics of knowledge while systems without that scheme are regarded as logics of belief.
\end{itemize} 

Observe that the notation is organized in the following way. In $ELn^-$, $ELn$, and in $EkLn^-$, $EkLn$ ($k=4,5$; $n=3,4,5$), the index $n$ refers to extensions by corresponding modal laws from S3, S4, S5, respectively, whereas the index $k$ refers to corresponding epistemic extensions. Exceptions from that rule are $E6Ln^-$ and $E6Ln$, where $n$ still refers to the corresponding modal extension, and number $6$ is chosen here to indicate the addition of the two bridge axioms (PNB) and (NNB) to $ELn^-$ and $ELn$, respectively.

The notion of derivation is defined in the usual way: Suppose $\mathcal{L}$ is one of our modal systems and $\Phi\cup\{\varphi\}$ is a subset of the corresponding object language: $Fm_1$ or $Fm$. We say that $\varphi$ is derivable from $\Phi$ in $\mathcal{L}$, notation: $\Phi\vdash_\mathcal{L}\varphi$, if there is a finite sequence $\varphi_1,...,\varphi_m=\varphi$ such that for each $\varphi_i$ ($i=1,...,m$), either is $\varphi_i$ an axiom of $\mathcal{L}$ or $\varphi_i\in\Phi$ or $\varphi_i$ is of the form (TND) or there is an axiom $\psi$ of $\mathcal{L}$ such that $\varphi_i=\square\psi$ (application of AN) or there are formulas $\varphi_j$, $\varphi_k=\varphi_j\rightarrow\varphi_i$ occurring in the sequence, where $j,k<i$ (application of MP).  

System $L$ was originally introduced in \cite{lewjlc2} as a minimal modal logic satisfying the condition `strict equivalence = propositional identity' and combining classical and intuitionistic propositional logic in the following sense: $L$ is a conservative extension of CPC and for any set of propositional formulas $\Phi\cup\{\varphi\}\subseteq Fm_0$ it holds that
\begin{equation}\label{30}
\Phi\vdash_{IPC}\varphi\text{ }\Leftrightarrow\text{ }\square\Phi\vdash_L\square\varphi.
\end{equation}

The stronger logic $L3$, first considered in \cite{lewigpl}, inherits these properties and distinguishes explicitly between intuitionistic and classical principles: while all axioms are supposed to be intuitionistically acceptable, the unique theorem scheme of \textit{tertium non datur} (TND) represents a classical law. The S3 principle (A3), contained in $L3$ as an axiom, ensures that all instances of SP are derivable. Actually, one can show a stronger fact: all instances of SP prefixed by $\square$ are derivable, i.e. $L3$ contains $\square$SP (see \cite{lewjlc1} for a proof where it is shown that S3 contains $\square$SP). Obviously, we get the following hierarchy: $L\subseteq L3\subseteq L4\subseteq L5$.

The systems $EL3^-$ and $EL3$--$EL5$ were introduced in \cite{lewigpl} as epistemic/modal extensions of $L3$. These classical systems seem to reflect in a sense the basic principles of Intuitionistic Epistemic Logic introduced in \cite{artpro}. The precise relationship between the S5-style modal systems of our hierarchy (i.e. those containing $L5$) and Intuitionistic Epistemic Logic will become explicit by means of the uniform framework of relational semantics presented below.

\begin{lemma}\label{90}
Our logics of belief form the following hierarchies:
\begin{itemize}
\item $EL3^-\subseteq EL4^-\subseteq EL5^-$
\item $ELn^-\subseteq E4Ln^-\subseteq E5Ln^- \subseteq E6Ln^-$, for $n=3,4,5$
\item $EkL3^-\subseteq EkL4^-\subseteq EkL5^-$, for $k=4,5,6$. 
\end{itemize}
Corresponding hierarchies hold for our logics of knowledge.
\end{lemma}

\begin{proof}
Most of the inclusions follow immediately from the definitions. It remains to prove $E5Ln^- \subseteq E6Ln^-$, for $n=3,4,5$. It is enough to show that $\square (K\varphi\rightarrow K K\varphi)$ and $\square (\neg K\varphi\rightarrow K \neg K\varphi)$ are theorems of $E6Ln^-$. Observe that $K\varphi\rightarrow\square K\varphi$ is an instance of (PNB), $\square K\varphi\rightarrow\square K K\varphi$ is an instance of (CoRe), and $\square KK\varphi\rightarrow KK\varphi$ is an instance of (A2). Then rule AN along with transitivity axiom \eqref{20} and rule MP yields $\square (K\varphi\rightarrow K K\varphi)$. On the other hand, $\neg K\varphi\rightarrow\square \neg K\varphi$ is an instance of (NNB),  $\square\neg K\varphi\rightarrow\square K\neg K\varphi$ is an instance of (CoRe) and $\square K\neg K\varphi\rightarrow K \neg K\varphi$ is an instance of (A2). In the same way as before, we derive $\square (\neg K\varphi\rightarrow K \neg K\varphi)$. 
\end{proof}

The full Necessitation Rule of normal modal logics is, in general, not applicable in our systems. However, the rule is valid in the `intuitionistic parts' of logics containing (A4) according to the following result.

\begin{lemma}\label{80}
Let $\mathcal{L}$ be a logic of our hierarchy containing axiom scheme (A4). If $\varphi$ is a theorem of $\mathcal{L}$ derivable without (TND), then $\square\varphi$ is a theorem of $\mathcal{L}$.
\end{lemma}

\begin{proof}
This is an induction on the length of a derivation. If axioms occur in the derivation, then apply AN. If rule MP occurs, then apply the induction hypothesis and the modal distribution law. If some formula $\psi$ in the derivation is obtained by rule AN, i.e. $\psi=\square\chi$, for some axiom $\chi$, then we derive $\square\psi$ by applying (A4) and rule MP. 
\end{proof}

Logics with axiom scheme (A4) can be axiomatized in a slightly simpler way:

\begin{corollary}\label{90}
Let $\mathcal{L}$ be a logic of our hierarchy containing (A4). Replacing the axiom schemes (A3) and co-reflection (CoRe) with the usual distribution law $\square(\varphi\rightarrow\psi)\rightarrow(\square\varphi\rightarrow\square\psi)$ and the weaker version $\square\varphi\rightarrow K\varphi$ of (CoRe), respectively, results in a system which is deductively equivalent with $\mathcal{L}$.
\end{corollary}

\begin{proof}
It is clear that those obviously weaker principles are theorems of $\mathcal{L}$. By Lemma \ref{80}, we also derive $\square(\square(\varphi\rightarrow\psi)\rightarrow(\square\varphi\rightarrow\square\psi))$ and $\square(\square\varphi\rightarrow K\varphi)$ (alternatively, we may derive these formulas applying the S1 transitivity principle \eqref{20} above). The other way round, we now suppose $\square(\varphi\rightarrow\psi)\rightarrow(\square\varphi\rightarrow\square\psi)$ and $\square\varphi\rightarrow K\varphi$ are given as axioms and show that (A3) and (CoRe), prefixed by $\square$, are derivable, provided (A4) is available. Rule AN yields $\square(\square\varphi\rightarrow K\varphi)$, and distribution along with MP yields $\square\square\varphi\rightarrow \square K\varphi$. Then, by (A4) $\square\varphi\rightarrow\square\square\varphi$ and transitivity of implication, we get (CoRe). On the other hand, AN applied to the distribution axiom, along with distribution itself and MP yields $\square\square(\varphi\rightarrow\psi)\rightarrow\square(\square\varphi\rightarrow\square\psi)$. By (A4) in the form of $\square(\varphi\rightarrow\psi)\rightarrow\square\square (\varphi\rightarrow\psi)$ along with transitivity of implication, we get (A3). Now, by the proof of Lemma \ref{80}, we also derive the respective formulas prefixed by $\square$, i.e. $\square(\square\varphi\rightarrow \square K\varphi)$ and $\square (\square(\varphi\rightarrow\psi)\rightarrow\square(\square\varphi\rightarrow\square\psi))$. 
\end{proof}

Next, we present some properties shared by all our modal logics. 

\begin{lemma}\cite{lewjlc1, lewigpl}\label{100}
The Substitution Principle SP, i.e. axiom (Id3), is valid in all our modal logics, even in the extended epistemic language. Consequently, (Id1)--(Id3) are satisfied and we have `propositional identity = strict equivalence' in all our modal logics.
\end{lemma}

\begin{proof}
Of course, SP holds in $L$ where it is explicitly stated as a theorem scheme of the deductive system. SP involves the notion of a substitution $\varphi[x:=\psi]$ which is defined in the canonical way by induction on the complexity of formula $\varphi$. This means in particular that the validity of SP depends on the underlying language. We showed in \cite{lewjlc1} that in the language of modal logic, all instances of SP are theorems of S3. By the same arguments given there, SP is also valid in $L3$. In \cite{lewigpl}, we showed that SP is also valid in $EL3^-$, a logic defined over the extended epistemic language. Then SP also holds in all extensions of $EL3^-$ which are defined over same language. 
\end{proof}

A certain substitution or replacement principle of CPC says that if two formulas $\varphi, \psi$ are logically equivalent, then replacing an occurrence of $\varphi$ with $\psi$ in a formula $\chi$ results in a formula which is logically equivalent to $\chi$. It is well-known that that principle is also valid in normal modal logics. It fails, however, in our modal logics which are of higher `intensional degree' than current classical modal logics in the sense that more propositions can be distinguished: classically equivalent formulas such as $\varphi$ and $\neg\neg\varphi$ denote the same proposition in any model of normal modal logic (in fact, $\square(\varphi\leftrightarrow\neg\neg\varphi)$ is a theorem), though such formulas generally denote distinct propositions in models of $L$ (only intuitionistically equivalent formulas have always the same denotation). Actually, we are working with non-Fregean logics: the Fregean Axiom $(\varphi\leftrightarrow\psi)\rightarrow (\varphi\equiv\psi)$, ``Formulas with the same truth value have the same meaning", is invalid.\footnote{Lewis modal logics S3--S5 are examples of current modal systems that can be studied as non-Fregean logics (see, e.g., \cite{lewjlc1, lewsl}).}

\begin{lemma}\label{110}
Let $\mathcal{L}$ be one of our modal logics. Then
\begin{equation*}
\vdash_\mathcal{L} \varphi\leftrightarrow\psi \text{ does not generally imply } \vdash_\mathcal{L}\chi[x:=\varphi]\leftrightarrow\chi[x:=\psi].
\end{equation*}
\end{lemma}

\begin{proof}
We anticipate here the fact that our logics are sound and complete w.r.t. the kind of algebraic semantics considered in \cite{lewigpl}. Towards a counterexample, we consider the model constructed in the proof of [Theorem 4.4, \cite{lewigpl}] which is a model of $EL5$ based on the linearly ordered Heyting algebra of the interval of reals $[0,1]$ with ultrafilter $(0,1]$. One easily recognizes that this is actually a model of $E6L5$, i.e. a model of all our modal logics. Of course, the formula $y\leftrightarrow\neg\neg y$ is a classical tautology and therefore valid in our classical modal logics. $y$ and $\neg\neg y$ have always the same classical truth value. However, these formulas may denote distinct propositions in some model of our non-Fregean logics. For instance, in the considered model, one finds an element $m\neq 1$ such that the double negation of $m$ equals $1$. In fact, every $m\in (0,1)$ has this property. So if the variable $y$ denotes such an $m$, then $\square y$ denotes $0$ and $\square\neg\neg y$ denotes $1$. Then that model together with $\varphi=y$, $\psi=\neg\neg y$ and $\chi =\square x$ represents a counterexample.
\end{proof}

By [Lemma 2.3, \cite{lewjlc1}], all biconditionals of the form $\square\varphi\leftrightarrow (\varphi\equiv\top)$ are valid in system S1+SP. We outline here a simpler proof of that fact.

\begin{theorem}\label{110} All biconditionals of the form 
\begin{equation}\label{40}
\square\varphi\leftrightarrow (\varphi\equiv\top)
\end{equation}
are theorems of our modal logics.
\end{theorem}

\begin{proof}
Of course, $\varphi\rightarrow\top$ is a theorem of IPC. AN then yields $\square(\varphi\rightarrow\top)$. On the other hand, $\varphi\rightarrow(\top\rightarrow\varphi)$ is the substitution-instance of a theorem of IPC and AN yields $\square(\varphi\rightarrow(\top\rightarrow\varphi))$. Applying distribution and MP, we get $\square\varphi\rightarrow\square(\top\rightarrow\varphi)$. Hence, $\square\varphi\rightarrow (\square(\varphi\rightarrow\top)\wedge \square(\top\rightarrow\varphi))$, i.e. $\square\varphi\rightarrow (\varphi\equiv\top)$ is derivable. The other way round, $(\varphi\equiv\top)$ implies in particular $\square(\top\rightarrow\varphi)$. By distribution and MP, $\square\top\rightarrow\square\varphi$. $\top$ is a theorem of IPC. Hence, AN and MP yield $\square\varphi$. This shows that $(\varphi\equiv\top)\rightarrow\square\varphi$ is derivable. 
\end{proof}

Suppose $\equiv$ is a primitive symbol of the language such that the axioms (Id1)--(Id3) of propositional identity are satisfied. Under this assumption, we called the scheme $(\square\varphi\wedge\square\psi)\rightarrow (\varphi\equiv\psi)$ the \textit{Collapse Axiom} in \cite{lewsl}. It implies that there is exactly one necessary proposition, namely the proposition denoted by $\top$. This is also expressed by scheme \eqref{40} above, where $\equiv$ is defined as strict equivalence. However, we cannot expect that strict equivalence is propositional identity in \textit{any} modal system. For instance, axiom (Id2) is not fulfilled in normal modal logic $K$, and (Id3) does not hold in Lewis systems S1 and S2 (see \cite{lewsl, lewjlc1}). Nevertheless, the relation of propositional identity refines the relation of strict equivalence.\footnote{We show that $\varphi\equiv\psi$ refines strict equivalence $\square(\varphi\leftrightarrow\psi)$. By (Id3), $(\varphi\equiv\psi)\rightarrow (\square(\varphi\leftrightarrow x)) [x:=\varphi]\equiv (\square (\varphi\leftrightarrow x))[x:=\psi])$. Thus, $(\varphi\equiv\psi)\rightarrow (\square (\varphi\leftrightarrow \varphi)\equiv \square (\varphi\leftrightarrow \psi))$. This along with (Id2) yields $(\varphi\equiv\psi)\rightarrow (\square (\varphi\leftrightarrow \varphi)\rightarrow \square (\varphi\leftrightarrow \psi))$. Since $\square (\varphi\leftrightarrow\varphi)$ is a theorem, also $(\varphi\equiv\psi)\rightarrow \square (\varphi\leftrightarrow \psi)$ is a theorem.} The Collapse Axiom is valid whenever both relations coincide. The logics S3--S5 as well as our modal logics are all strong enough to ensure that strict equivalence is propositional identity axiomatized by (Id1)--(Id3).\\
Thus, \eqref{40} says that there is exactly one necessary proposition, denoted by $\top$. In our modal logics, this is `the intuitionistically true proposition'. Then \eqref{40} reads
\begin{equation*}
\textit{``$\square\varphi$ iff $\varphi$ holds intuitionistically"}.
\end{equation*}
Under this interpretation, \eqref{40} recalls the \textit{Tarski biconditionals} (also called Convention T or T-scheme) of Tarski's truth theory. In our setting, operator $\square$ is a predicate for \textit{intuitionistic} truth and belongs to the object language.\footnote{Note, however, that the expression $\varphi\equiv\top$ is defined in terms of $\square$, i.e. the truth predicate appears on both sides of the biconditional \eqref{40}. This could be possibly avoided by introducing the identity connective as a primitive symbol along with a suitable additional axiomatization.} In fact, it will follow from the definition of our model-theoretic semantics, which is based on Heyting algebras, that $\square\varphi$ is satisfied (i.e. classically true) in a given model iff $\varphi$ denotes the top element of the underlying Heyting lattice (i.e. $\varphi$ is intuitionistically true in the model).\\

In the following, we present a few examples of derivable principles. 

\begin{theorem}\label{130}
\noindent
\begin{enumerate}
\item $\vdash_{L}\neg\Diamond\bot$. ``A proof of $\bot$ is impossible ($\bot$ has no possible proof)." ``The set of proved propositions is consistent".
\item $\vdash_{L}\varphi\rightarrow\Diamond\varphi$. ``Classical truth implies the possibility of a proof."
\item $\vdash_{L}\square\varphi\rightarrow\Diamond\varphi$. ``If $\varphi$ has an actual proof, then $\varphi$ has a possible proof."
\item $\vdash_{L4}\square\neg\Diamond\bot$. ``There is an actual proof that a proof of $\bot$ is impossible." ``There is an actual proof that the given set of proved propositions is consistent."
\item $\vdash_{L5}\Diamond\square\varphi\rightarrow\square\square\varphi$. ``If $\square\varphi$ has a possible proof, then $\square\varphi$ has an actual proof.''
\item $\vdash_{L5}\Diamond\neg\square\varphi\rightarrow\square\neg\square\varphi$. ``If $\neg\square\varphi$ has a possible proof, then it has an actual proof."
\item $\vdash_{L5}\square(\square\varphi\vee\neg\square\varphi)$. ``There exists an actual proof of the fact that either there is an actual proof of $\varphi$ or there is no such proof of $\varphi$."
\end{enumerate}
\end{theorem}

\begin{proof}
(i): Formula $\neg\bot = \bot\rightarrow\bot$ is an intuitionistic tautology. By AN, $\square\neg\bot$ is a theorem. Also $\square\neg\bot\rightarrow\neg\neg\square\neg\bot$ is a theorem (a substitution-instance of the intuitionistic tautology $x\rightarrow\neg\neg x$). By MP, we derive $\neg\neg\square\neg\bot = \neg\Diamond\bot$. \\
(ii): By (A2), $\square\neg\varphi\rightarrow\neg\varphi$. Apply contraposition.\\
(iii): Consider (A2), item (ii) and transitivity of implication.\\
(iv): Note that (TND) does not occur in the derivation given in (i). Thus, by Lemma \ref{80}, $\square\neg\Diamond\bot$ is a theorem of $L4$. \\
(v): $\Diamond\square\varphi\rightarrow\square\varphi$ is the contrapositive of (A5). This together with (A4) and transitivity of implication yields $\Diamond\square\varphi\rightarrow\square\square\varphi$.\\
(vi): $\square\varphi\rightarrow\neg\neg\square\varphi$ is the substitution-instance of a theorem of IPC. Rule AN, distribution and MP yield $\square\square\varphi\rightarrow\square\neg\neg\square\varphi$. Then by (A4) and transitivity of implication, $\square\varphi\rightarrow\square\neg\neg\square\varphi$. The contrapositive is $\Diamond\neg\square\varphi\rightarrow\neg\square\varphi$. Finally, $\Diamond\neg\square\varphi\rightarrow\square\neg\square\varphi$ by (A5) and transitivity of implication.\\
(vii): $\chi\rightarrow (\chi\vee\psi)$ is an axiom. By AN and distribution, 
\begin{equation*}
\square\chi\rightarrow \square(\chi\vee\psi)
\end{equation*}
is a theorem scheme of which $\square\square\varphi\rightarrow\square (\square\varphi\vee\neg\square\varphi)$ is an instance. Axiom (A4) along with transitivity of implication then yields $\square\varphi\rightarrow\square (\square\varphi\vee\neg\square\varphi)$. $\square\neg\square\varphi\rightarrow\square(\neg\square\varphi\vee\square\varphi)$ is a further instance of the theorem scheme above. This, together with (A5), yields $\neg\square\varphi\rightarrow \square(\neg\square\varphi\vee\square\varphi)$. Since $(\chi\vee\neg\chi)\equiv (\neg\chi\vee\chi)$ is a theorem, $(\chi\vee\neg\chi)$ and $(\neg\chi\vee\chi)$ are interchangeable in every context (Substitution Principle). Thus,
\begin{equation*}
\square\varphi\rightarrow\square (\square\varphi\vee\neg\square\varphi)\text{ and }\neg\square\varphi\rightarrow \square(\square\varphi\vee\neg\square\varphi).
\end{equation*}
are theorems. Consequently, $(\square\varphi\vee\neg\square\varphi)\rightarrow \square(\square\varphi\vee\neg\square\varphi)$ is a theorem. \textit{Tertium non datur} along with MP then yields the assertion.
\end{proof}

Intuitively, $K\varphi$ has a possible proof, notation: $\Diamond K\varphi$, if and only if classical truth of $K\varphi$ is possible. If $K\varphi$ is classically true, then $\varphi$ denotes a believed (a known) proposition. In this sense, we may read $\Diamond K\varphi$ as ``$\varphi$ is believable (knowable)". We now list some examples of derivable (non-derivable) epistemic principles. Some of these examples are given in similar form in [Theorem 2.5, \cite{lewigpl}].

\begin{theorem}\label{140}
\noindent
\begin{enumerate}
\item $\vdash_{EL3^-}K\varphi\rightarrow\Diamond K\varphi$. ``Believed (known) propositions are believable (knowable), respectively."
\item $\vdash_{EL4^-} \square\varphi\rightarrow \square K\square\varphi$. ``If $\varphi$ has an actual proof, then there is an actual proof that it is believed (known) that $\varphi$ has an actual proof." 
\item $\vdash_{EL5^-} \neg\square\varphi\rightarrow \square K\neg\square\varphi$. ``If $\varphi$ has no actual proof, then there is an actual proof that it is believed (known) that $\varphi$ has no actual proof." 
\item $\vdash_{EL5^-} K\square\varphi\vee K\neg\square\varphi$. ``Either it is believed (known) that $\varphi$ has an actual proof or it is believed (known) that $\varphi$ has no such proof, respectively."
\item $\vdash_{E6L3^-}\Diamond K\varphi\rightarrow K\varphi$. ``All believable (knowable) propositions are believed (known), respectively."
\item $\nvdash_{E6L5}\varphi\rightarrow\Diamond K\varphi$. ``There may exist true propositions that are unknowable (unbelievable)."
\item $\vdash_{EL5^-}(K\varphi\rightarrow\square K\varphi)\rightarrow (\neg K\varphi\rightarrow\square\neg K\varphi)$. ``Modulo $EL5^-$, axiom scheme (NNB) follows from (PNB)."
\end{enumerate}
\end{theorem}

\begin{proof}
(i): This is a particular case of item (ii) of the preceding Theorem.\\
(ii): $\square\varphi\rightarrow\square\square\varphi$ and $\square\square\varphi\rightarrow\square K\square\varphi$ are axioms of $EL4^-$. The assertion follows by transitivity of implication.\\
(iii): $\neg\square\varphi\rightarrow\square\neg\square\varphi$ and $\square\neg\square\varphi\rightarrow\square K\neg\square\varphi$ are axioms of $EL5^-$. Apply transitivity of implication.\\
(iv): This is shown in [Theorem 2.6, \cite{lewigpl}].\\ 
(v): This is the contrapositive of axiom (NNB).\\
(vi): Again, we anticipate algebraic semantics given below and prove the assertion by presenting an $E6L5$-model where $\varphi$ denotes a true and $\Diamond K\varphi$ denotes a false proposition. Such a model should contain a proposition $m$ such that $m$ is true and unknown. If variable $x$ denotes $m$, the model then satisfies $x$ and $\neg K x$. Since the model validates axiom (NNB), $\square \neg K x$ is true, i.e. $\Diamond Kx$ is false. In the proof of [Theorem 4.4, \cite{lewigpl}], a specific model of logic $EL5$ is constructed. One easily checks that that model is a model of logic $E6L5$ (in the sense of the semantics presented below) and has the desired properties.\\
(vii): By (PNB), (A2) and rule AN, we have $K\varphi\equiv\square K\varphi$. Then by SP, the formulas $K\varphi$ and $\square K\varphi$ can be replaced by each other in every context. In particular, $\neg K\varphi\rightarrow \neg\square K\varphi$ is a theorem. By (A5), $\neg\square K\varphi\rightarrow \square\neg\square K\varphi$ is a theorem. Replacing in this formula the last occurrence of $\square K\varphi$ with $K\varphi$, we get $\neg\square K\varphi\rightarrow \square\neg K\varphi$. Transitivity of implication yields $\neg K\varphi\rightarrow \square \neg K\varphi$.
\end{proof}

Note that item (ii) of Theorem \ref{140} is related to the axiom of intuitionistic co-reflection $\varphi\rightarrow K\varphi$ of IEL. The classical version of that axiom under G\"odel translation ``box every subformula" is $\square (\square\varphi' \rightarrow \square K \square\varphi')$, where $\varphi'$ is the translation of $\varphi$. By Lemma \ref{80} and Theorem \ref{140} (ii), $\square (\square\varphi\rightarrow \square K\square\varphi)$ is a theorem of $EL4^-$. That is, the classical interpretation of intuitionistic co-reflection, based on G\"odel translation, is similar to a theorem of $EL4^-$. However, our modal version of co-reflection, axiom (CoRe) $\square\varphi\rightarrow \square K\varphi$, is strictly weaker than intuitionistic co-reflection as the uniform framework of relational semantics below will reveal. IPC is contained in our classical modal logics, though IEL is not. In the last section, we will present an extension of $EL5$ that contains IEL in a similar way as $L$ contains IPC.\\

Our standard semantics is based on the model-theoretic, algebraic semantics introduced in \cite{lewjlc2, lewigpl}, where models are given as Heyting algebras with a designated ultrafilter and some additional structure. Before presenting the details, we list some basic facts about Heyting algebras (Heyting lattices). Recall that a filter of a lattice $(M,\le)$ is a non-empty subset $F\subseteq M$ such that the following two conditions are satisfied for all $m,m'\in M$: if $m,m' \in F$, then $f_\wedge(m,m')\in F$; if $m\in F$ and $m\le m'$, then $m'\in F$. A filter is said to be proper if it does not contain all elements of the underlying lattice. An ultrafilter is a maximal proper filter; a proper filter $F$ is said to be prime if $f_\vee(m,m')\in F$ implies [$m\in F$ or $m'\in F$]. Note that we use symbols of the form $f_\vee$, $f_\wedge$, $f_\rightarrow$ etc. to refer in an intuitive way to lattice operations such as join, meet, relative pseudo-complement etc.

\begin{lemma}\label{200}
Let $\mathcal{H}$ be a Heyting algebra with universe $H$.\\
(a) $U\subseteq H$ is an ultrafilter iff there is a Heyting algebra homomorphism $h$ from $\mathcal{H}$ to the two-element Boolean algebra $\mathcal{B}$ such that the top element of $\mathcal{B}$ is precisely the image of $U$ under $h$.\\
(b) If $U\subseteq H$ is an ultrafilter, then for all $m,m'\in H$:
\begin{itemize}
\item $f_\vee(m,m')\in U$ iff $m\in U$ or $m'\in U$ (i.e. $U$ is a prime filter)
\item $m\in U$ or $f_\neg(m)\in U$
\item $f_\rightarrow(m,m')\in U$ iff [$m\notin U$ or $m'\in U$] iff $f_\vee(f_\neg(m),m')\in U$.
\end{itemize}
(c) Every proper filter is the intersection of all prime filters containing it.\\
(d) Let $m_1,m_2\in H$ and let $P$ be a prime filter. Then we have $f_\rightarrow(m_1,m_2)\in P$ if, and only if, for all prime filters $P'\supseteq P$, $m_1\in P'$ implies $m_2\in P'$.
\end{lemma}

\begin{proof}
(a)--(d) are known properties of Heyting algebras which are not hard to prove. The right-to-left implication of (d), however, might be less familiar. We outline a proof.
Let $P$ be a prime filter. We consider the quotient Heyting algebra $\mathcal{H}'$ of $\mathcal{H}$ modulo $P$. That is, the elements of $\mathcal{H'}$ are the equivalence classes $\overline{m}$ of $m\in M$ modulo the equivalence relation $\sim$ defined by $m\sim m'\Leftrightarrow$ [$f_\rightarrow(m,m')\in P$ and $f_\rightarrow(m',m)\in P$]. Then one easily checks that $P$ is the equivalence class of $f_\top$ modulo $\sim$, and it is the top element $f'_\top$ of $\mathcal{H'}$.\\
\textbf{Claim1}: Let $m,m'\in H$. If $\overline{m}\in F'$ implies $\overline{m'}\in F'$, for all filters $F'$ of $\mathcal{H}'$, then $\overline{m}\le' \overline{m'}$, where $\le'$ is the lattice order of $\mathcal{H}'$.\\
\textit{Proof of Claim1}. Suppose $\overline{m}\nleqslant' \overline{m'}$. Consider the filter $G=\{\overline{m''}\mid \overline{m}\le' \overline{m''}\}$. Then $\overline{m}\in G$ and $\overline{m'}\notin G$. We have proved the Claim.\\
\textbf{Claim2}: Let $m,m'\in H$. If $\overline{m}\in F'$ implies $\overline{m'}\in F'$, for all prime filters $F'$ of $\mathcal{H}'$, then $\overline{m}\le' \overline{m'}$, where $\le'$ is the lattice ordering of $\mathcal{H}'$.\\
\textit{Proof of Claim2}. Claim2 follows from Claim1 together with (c).\\
\textbf{Claim3}: If $F'$ is a (prime) filter of $\mathcal{H}'$, then $F=\{m\mid\overline{m}\in F'\}$ is a (prime) filter of $\mathcal{H}$ extending $P$.\\
\textit{Proof of Claim3}. Suppose $m\in F$ and $m\le m'$. Then $f_\rightarrow(m,m')=f_\top$. Thus, $\overline{f_\rightarrow(m,m')}=P=f'_\top$. That is, $f'_\rightarrow(\overline{m},\overline{m'})=f'_\top$ and therefore $\overline{m}\le'\overline{m'}$. It follows that $\overline{m'}\in F'$ and $m'\in F$. The remaining filter properties follow straightforwardly. $m\in P$ implies $\overline{m}=P=f'_\top\in F'$ implies $m\in F$. Thus, $P\subseteq F$ and Claim3 holds true.\\
Now let $m_1, m_2$ be elements of $\mathcal{H}$ such that for all prime filters $F\supseteq P$, $m_1\in F$ implies $m_2\in F$. We show that this implies $f_\rightarrow(m_1,m_2)\in P$. Let $\overline{m_1}\in F'$ for some prime filter $F'$ of $\mathcal{H}'$. Then, by Claim3, $m_1\in F=\{m\mid\overline{m}\in F'\}$ and $F$ is a prime filter of $\mathcal{H}$ with $P\subseteq F$. By hypothesis, $m_2\in F$. Thus, $\overline{m_2}\in F'$. By Claim2, $\overline{m_1}\le' \overline{m_2}$. Then $\overline{f_\rightarrow(m_1,m_2)}=f'_\top=P$. That is, $f_\rightarrow(m_1,m_2)\in P$.
\end{proof}

For further details about Heyting algebras (sometimes called pseudo-Boolean algebras) we refer the reader to \cite{chazak}.
 
We would like to point out that dropping the epistemic ingredients of the following model definitions results in models for the logics $L3$--$L5$, i.e. an $EL3^-$-model becomes a $L3$-model, etc. Since these transitions are trivial, we only consider here the more complex \textit{epistemic} models and do not treat semantics (neither completeness) of $L3$--$L5$ separately. Algebraic semantics and completeness of our weakest modal logic $L$ is established in \cite{lewjlc2}.

\begin{definition}\cite{lewigpl}\label{205}
An $EL3^-$-model, to which we also refer as an epistemic model, is a Heyting algebra 
\begin{equation*}
\mathcal{M}=(M, \mathit{TRUE}, \mathit{BEL}, f_\bot, f_\top, f_\vee, f_\wedge, f_\rightarrow, f_\square, f_K)
\end{equation*}
with universe $M$, a designated ultrafilter $\mathit{TRUE}\subseteq M$, a set $\mathit{BEL}\subseteq M$ and additional unary operations $f_\square$ and $f_K$ such that for all $m,m'\in M$ the following truth conditions are fulfilled (as before, $\le$ denotes the lattice order):
\begin{enumerate}
\item $f_\square(f_\vee(m,m'))\le f_\vee(f_\square(m),f_\square(m'))$
\item $f_\square(m)\le m$
\item $f_\square(f_\rightarrow(m,m'))\le f_\square(f_\rightarrow(f_\square(m),f_\square(m')))$
\item $f_\square(m)\in\mathit{TRUE}\Leftrightarrow m=f_\top$
\item $f_K(m)\in\mathit{TRUE}\Leftrightarrow m\in\mathit{BEL}$
\item $f_K(f_\rightarrow(m,m'))\le f_\rightarrow(f_K(m),f_K(m'))$
\item $f_\square(m)\le f_\square(f_K(m))$
\end{enumerate}
$M$ is the universe of all \textit{propositions} and $\mathit{TRUE}\subseteq M$ is the subset of classically true propositions. The propositions $f_\top$, $f_\bot$ represent intuitionistic truth and intuitionistic falsity, respectively. $\mathit{BEL}$ is the set of `believed' propositions. A proposition $m$ is said to be `known' if $m\in\mathit{TRUE}\cap\mathit{BEL}$.\footnote{We refer to the elements of $\mathit{TRUE}$, i.e. the true propositions, also as \textit{facts}. If $\mathit{BEL}\subseteq\mathit{TRUE}$, then all believed propositions are facts and belief becomes knowledge.} 
\end{definition}

Note that the set $\mathit{BEL}$ and condition (v) are redundant data in the definition. We could drop them and restore the set $\mathit{BEL}$ defining $\mathit{BEL}:=\{m\in M\mid f_K(m)\in\mathit{TRUE}\}$. Nevertheless, we keep these extra data in the definition in order to explicitly indicate that \textit{belief} and \textit{knowledge} are extensionally modeled as sets of propositions, similarly as \textit{classical truth}. Knowledge is given as a set of facts. In the sense of (vi), belief is `closed under Modus Ponens'. \\

We shall tacitly make use of the equivalence $m\le m' \Leftrightarrow f_\rightarrow(m,m')=f_\top$ which holds in all Heyting algebras. Note that truth conditions (i) and (iv) ensure that every model has the \textbf{Disjunction Property}: for all $m,m'\in M$, $f_\vee(m,m')=f_\top$ iff $m=f_\top$ or $m'=f_\top$. That is, the smallest lattice filter $\{f_\top\}$ is a prime filter.

\begin{definition}\label{210}
Let $\mathcal{M}$ be an $EL3^-$-model. We say that 
\begin{itemize}
\item $\mathcal{M}$ is an $EL4^-$-model if for all $m\in M$: $f_\square(m)\le f_\square(f_\square(m))$.
\item $\mathcal{M}$ is an $EL5^-$-model if for all $m\in M$: 
\begin{equation*}
\begin{split}
f_\square(m)=
\begin{cases}
f_\top, \text{ if }m=f_\top\\
f_\bot, \text{ else}
\end{cases}
\end{split}
\end{equation*}
\item $\mathcal{M}$ is an $E4Ln^-$-model, for $n=3,4,5$, if $\mathcal{M}$ is an $ELn^-$-model and for all $m\in M$: $f_K(m)\le f_K(f_K(m))$.
\item $\mathcal{M}$ is an $E5Ln^-$-model, $n=3,4,5$, if $\mathcal{M}$ is an $E4Ln^-$-model and for all $m\in M$: $f_\neg(f_K(m))\le f_K(f_\neg(f_K(m)))$.  
\item $\mathcal{M}$ is an $E6Ln^-$-model, $n=3,4,5$, if $\mathcal{M}$ is an $ELn^-$-model  and for all $m\in M$: $f_K(m) \le f_\square(f_K(m))$ and $f_\neg(f_K(m))\le f_\square(f_\neg(f_K(m)))$.
\item If $\mathcal{M}$ is an $ELn^-$-model or an $EkLn^-$-model, for $k\in\{4,5,6\}$ and $n\in\{3,4,5\}$, and $\mathcal{M}$ satisfies the additional truth condition 
\begin{equation}\label{50}
f_K(m)\le f_\neg(f_\neg(m))\text{ for all propositions }m,
\end{equation}
then we omit the superscript $^-$ in the notation and refer to $\mathcal{M}$ as an $ELn$-model or an $EkLn$-model, respectively. We refer to models with that additional truth condition also as models of knowledge.
\end{itemize}
Instead of ``$\mathcal{M}$ is an $EL3$-model" we also simply say ``$\mathcal{M}$ is $EL3$" (and similarly for the other classes of models).
\end{definition}

Note that in an $E6Ln^-$-model, for any proposition $m$ there are only two possibilities for $f_K(m)$: either $f_K(m)=f_\top$ or $f_K(m)=f_\bot$. In fact, $m\in\mathit{BEL}$ $\Rightarrow$ $f_K(m)\in\mathit{TRUE}$ $\Rightarrow$ $f_\square(f_K(m))\in\mathit{TRUE}$ $\Rightarrow$ $f_K(m)=f_\top$; and on the other hand, $m\notin\mathit{BEL}$ $\Rightarrow$ $f_\neg(f_K(m))\in\mathit{TRUE}$ $\Rightarrow$ $f_\square(f_\neg(f_K(m)))\in\mathit{TRUE}$ $\Rightarrow$ $f_\neg(f_K(m))=f_\top$ $\Rightarrow$ $f_K(m)=f_\bot$. One easily verifies that every $EL5^-$-model is an $EL4^-$-model, every $E5Ln^-$-model is an $E4Ln^-$-model and every $E6Ln^-$-model is an $E5Ln^-$-model.

\begin{lemma}\label{212}
Let $\mathcal{M}$ be a model. The set $\mathit{BEL}\subseteq M$ of believed propositions is a filter of the underlying Heyting algebra. If $\mathcal{M}$ is a model of knowledge, i.e. $\mathcal{M}$ satisfies truth condition \eqref{50} above, then $\mathit{BEL}\subseteq\mathit{TRUE}$ and $\mathit{BEL}$ is a proper filter.
\end{lemma}

\begin{proof}
By the truth conditions (iv), (vii), (ii) and (v) of a model, $f_\top\in\mathit{BEL}$. By (vi), for any $m,m'\in M$: if $f_\rightarrow(m,m')\in\mathit{BEL}$ and $m\in\mathit{BEL}$, then $m'\in\mathit{BEL}$. These two conditions are sufficient for $\mathit{BEL}$ being a filter of the underlying Heyting algebra (see, e.g., \cite{chazak}). Finally, suppose $\mathcal{M}$ is a model of knowledge, i.e. $f_K(m)\le f_\neg(f_\neg(m))$ for all $m\in M$. Applying Lemma \ref{200}, we have $m\in\mathit{BEL}$ $\Rightarrow$ $f_K(m)\in\mathit{TRUE}$ $\Rightarrow$ $f_\neg(f_\neg(m))\in\mathit{TRUE}$ $\Rightarrow$ $m\in\mathit{TRUE}$. Hence, $\mathit{BEL}\subseteq\mathit{TRUE}$. Since $\mathit{TRUE}$ is a proper filter, $\mathit{BEL}$ is proper, too.  
\end{proof}

There is a simple characterization of $EL5^-$- and $EL5$-models: 

\begin{lemma}\label{215}
Let $\mathcal{M}=(M, \mathit{TRUE}, f_\bot, f_\top, f_\vee, f_\wedge, f_\rightarrow, f_\square, f_K)$ be a Heyting algebra with ultrafilter $\mathit{TRUE}$ and additional unary operations $f_\square$ and $f_K$. Then $\mathcal{M}$ is an $EL5^-$-model iff the following conditions are satisfied for all $m, m'\in M$:\\
(a) $\mathcal{M}$ has the Disjunction Property\\
(b)
\begin{equation*}
\begin{split}
f_\square(m)=
\begin{cases}
f_\top, \text{ if }m=f_\top\\
f_\bot, \text{ else}
\end{cases}
\end{split}
\end{equation*}
(c) $f_K(f_\rightarrow(m,m'))\le f_\rightarrow(f_K(m),f_K(m'))$\\
(d) $f_K(f_\top)=f_\top$.\\
\noindent $\mathcal{M}$ is an $EL5$-model iff in addition to (a)--(d), the condition\\
(e) $f_K(m)\le f_\neg (f_\neg(m))$\\
is satisfied for all $m\in M$.
\end{lemma}   

\begin{proof}
The properties (a)--(d) follow easily from the definition of an $EL5^-$-model. The other way round, suppose the Heyting algebra $\mathcal{M}$ satisfies (a)--(d). We show that $\mathcal{M}$ is an $EL3^-$-model in the sense of Definition \ref{205}. Because of (b), $\mathcal{M}$ then is an $EL5^-$-model. We check the conditions (i)--(vii) of Definition \ref{205}. Condition (iv) follows immediately from (b), condition (v) is redundant if we define $\mathit{BEL}$ in the obvious way, and (vi) is given by (c). Note that for each $m\in M$, we have either $f_\square(m)=f_\bot$ or $f_\square(m)=f_\top$. In order to verify the remaining conditions (i)--(iii) and (vii), we may assume that the left hand side of each of those inequalities is given by the top element $f_\top$. Inequality (i) then follows from the fact that $\mathcal{M}$ has the Disjunction Property. Inequality (ii) follows readily. (iii) follows from the fact that $m\le m'$. By (d), $f_\square(f_K(f_\top))=f_\top$. Now, inequality (vii) follows. Finally, the last affirmation regarding $EL5$-models follows readily from the definition. 
\end{proof}

Recall that the models of $L3$--$L5$ are defined by dropping the epistemic ingredients of the models of $EL3^-$--$EL5^-$, respectively. Lemma \ref{215} then gives a very simple characterization of $L5$-models: The class of $L5$-models is given by all Heyting algebras with a designated ultrafilter and a modal operator $f_\square$ such that the Disjunction Property and condition (b) are satisfied. 

\begin{definition}\label{220}
An assignment in a model $\mathcal{M}$ is a function $\gamma\colon V\rightarrow M$ which extends in the canonical way to a function $\gamma\colon Fm\rightarrow M$. More specifically, we have $\gamma(\bot)=f_\bot$, $\gamma(\top)=f_\top$, $\gamma(\square\varphi)=f_\square(\gamma(\varphi))$, $\gamma(K\varphi)=f_K(\gamma(\varphi))$ and $\gamma(\varphi*\psi)=f_*(\gamma(\varphi),\gamma(\psi))$, for $*\in\{\vee,\wedge,\rightarrow\}$. If $\mathcal{L}$ is one of our modal logics, then an $\mathcal{L}$-interpretation is a tuple $(\mathcal{M},\gamma)$ consisting of a $\mathcal{L}$-model and an assignment $\gamma\in M^V$. The relation of satisfaction is defined by 
\begin{equation*}
(\mathcal{M},\gamma)\vDash\varphi :\Leftrightarrow \gamma(\varphi)\in\mathit{TRUE}.
\end{equation*}
If $(\mathcal{M},\gamma)\vDash\varphi$, then we say that $\varphi$ is true in $\mathcal{M}$ under assignment $\gamma\in M^V$. If $\varphi$ is true in $\mathcal{M}$ under all assignments $\gamma\in M^V$, then we write $\mathcal{M}\vDash\varphi$ and say that $\varphi$ is valid in $\mathcal{M}$. A formula $\varphi$ is valid in logic $\mathcal{L}$ if $\varphi$ is valid in all $\mathcal{L}$-models. The defined notions extend in the usual way to sets of formulas.\\
Logical consequence in logic $\mathcal{L}$ is defined as usual: 
\begin{equation*}
\Phi\Vdash_\mathcal{L}\varphi :\Leftrightarrow \text{ if }(\mathcal{M},\gamma)\vDash\Phi\text{ then }(\mathcal{M},\gamma)\vDash\varphi,
\end{equation*}
for any $\mathcal{L}$-interpretation $(\mathcal{M},\gamma)$.
\end{definition}

Recall that we are dealing with special non-Fregean logics in the sense that the identity connective $\equiv$, defined by strict equivalence, satisfies Suszko's identity axioms (Id1)--(Id3). The intended semantics of the identity connective is `identity of meaning': we read $\varphi\equiv\psi$ as ``$\varphi$ and $\psi$ have the same meaning" or ``$\varphi$ and $\psi$ denote the same proposition". The following result (see also \cite{lewjlc1, lewsl}) corresponds exactly to that intuition:

\begin{lemma}\label{222}
\begin{equation*}
(\mathcal{M},\gamma)\vDash\varphi\equiv\psi\Leftrightarrow\gamma(\varphi)=\gamma(\psi).
\end{equation*}
\end{lemma}

\begin{proof}
Suppose $\gamma(\varphi)=m$ and $\gamma(\psi)=m'$. Then $\gamma(\varphi\equiv\psi)\in\mathit{TRUE}$ iff $f_\square(f_\rightarrow(m,m'))\in\mathit{TRUE}$ and $f_\square(f_\rightarrow(m',m))\in\mathit{TRUE}$ iff $f_\rightarrow(m,m')=f_\top$ and $f_\rightarrow(m',m)=f_\top$ iff $m\le m'$ and $m'\le m$ iff $m=m'$.
\end{proof}

Soundness and completeness of logics $L$ and $EL5$ w.r.t. to corresponding algebraic semantics is shown in \cite{lewjlc2} and \cite{lewigpl}, respectively. Those proofs can be adapted straightforwardly to completeness results for the remaining modal logics of our hierarchies.

\begin{theorem}[Strong completeness]\label{600}
Suppose $\Phi\cup\{\varphi\}\subseteq Fm$ and let $\mathcal{L}$ be any of our modal logics. Then $\Phi\vdash_\mathcal{L}\varphi\Leftrightarrow\Phi\Vdash_\mathcal{L}\varphi$.
\end{theorem}

We finish this section with a discussion on \textbf{self-referential propositions}. Having in mind the intended meaning of the identity connective (see Lemma \ref{222}), we are able to express self-referential statements by means of equations (this kind of modeling self-reference was proposed in \cite{str} and subsequently used in, e.g., \cite{zei, lewsl0, lewigpl0}). For instance, the equation 
\begin{equation}\label{300}
x\equiv (x\rightarrow\bot)
\end{equation}
defines a version of the \textit{liar proposition}. In fact, if the equation is satisfied in a given model, then the proposition denoted by $x$ says ``This proposition implies the \textit{absurdum}" or ``This proposition is false" or ``I'm lying". Fortunately, equations defining such paradoxical self-referential statements are unsatisfiable (for essentially the same reasons as $\varphi\leftrightarrow\neg\varphi$ is unsatisfiable in two-valued classical logic). The liar proposition, as a semantic object, does not exist.\\

Is there a proposition saying ``I am proved"? Since we identify \textit{proof} with \textit{intuitionistic truth}, we are actually asking for the existence of a \textit{truth-teller proposition} ``I am true", where we mean intuitionistic truth. This is a non-paradoxical proposition which can be defined by the equation 
\begin{equation}\label{310}
x\equiv \square x.
\end{equation}
Any proposition $m$ that solves the equation is an intuitionistic truth-teller saying ``I am proved". A truth-teller $m$ may be classically true (i.e. $m\in\mathit{TRUE}$) or classically false (i.e. $m\in M\smallsetminus \mathit{TRUE}$). In our specific example, there is only one potential true truth-teller, namely the top element $f_\top$ of a model (because of truth condition (iv) of Definition \ref{205}). The top element of each $L4$-model is a truth-teller, since $f_\square(m)=f_\top$ $\Leftrightarrow$ $m=f_\top$. On the other hand, in every $L5$-model, the proposition $f_\bot$ is a false truth-teller. In models which are not $L5$, there might exist further truth-tellers that are classically false, i.e. fixed points of $f_\square$ not belonging to $\mathit{TRUE}$.\\

Recall that we may read $\Diamond\varphi$ as ``$\varphi$ is consistent (with the given set of proved propositions)". Is there a proposition asserting its own consistency? We are asking for a solution of the equation
\begin{equation}\label{320}
x\equiv \Diamond x.
\end{equation}
One easily checks that in any $L4$-model, the bottom element $f_\bot$ is a solution; and in an $L5$-model, the top element $f_\top$ is a solution as well. So \eqref{320} is a further example of an equation allowing both true and false propositions as solutions in suitable models. There might exist further solutions distinct from $f_\top$ and $f_\bot$. Finally, a proposition that asserts its own inconsistency is described by the equation
\begin{equation}\label{330}
x\equiv \neg\Diamond x.
\end{equation}
An alternative version of that self-referential statement is given by the equation
\begin{equation}\label{340}
x\equiv \square\neg x.
\end{equation}
The difference between both equations is subtle. In fact, $\neg\Diamond x =\neg\neg\square\neg x$ and $\square\neg x$ are logically equivalent formulas in our classical modal logics. From an intuitionistic point of view, however, they express different intensions and therefore may denote different propositions. Suppose equation \eqref{340} is true in a given model. Then the proposition $m$ denoted by $x$ says ``There is an actual proof that I am false" or, in other words, ``I am inconsistent with the proved propositions". Suppose $m$ is classically true, i.e. $m=f_\square(f_\neg((m))\in\mathit{TRUE}$, then, by truth condition (iv) of a model, $f_\neg(m)=f_\top$, i.e. $m=f_\bot\notin\mathit{TRUE}$. This contradiction shows that $m$ cannot be classically true. So whenever equation \eqref{340} is satisfied, the proposition $m$ denoted by $x$ must be classically false. Furthermore, $m$ cannot be the proposition $f_\bot$ for otherwise $m=f_\square(f_\neg(m))=f_\square(f_\neg(f_\bot))=f_\square(f_\top)\in\mathit{TRUE}$, contradicting $m\notin\mathit{TRUE}$. We conclude that $m$ is a false proposition distinct from $f_\bot$. In particular, $m$ has a possible proof (i.e. $m$ is consistent with the proved propositions). Hence, $m$ is a consistent proposition asserting its own inconsistency. This sounds paradoxical. Note, however, that $m$ is classically false -- so there is no paradox. As a non-paradoxical proposition, a solution of \eqref{340} should exist in some model. Interestingly, since $m\notin\{f_\top, f_\bot\}$, no $L5$-model satisfies \eqref{340}. 

The simplest self-referential statements involving knowledge are described by the equations
\begin{equation}\label{350}
x\equiv Kx
\end{equation}
\begin{equation}\label{360}
x\equiv \neg Kx.
\end{equation}
Obviously, if  \eqref{350} is true, then the proposition denoted by $x$ says ``I am known", and if  \eqref{360} is true, then $x$ denotes a proposition saying ``I am unknown". Equation \eqref{350} is satisfied in every model where $x$ denotes $f_\top$. Consider equation \eqref{360} and assume that $K$ stands for knowledge as true belief, i.e. scheme $K\varphi\rightarrow\varphi$ is valid. Then a solution must be a proposition that is classically true and unknown. If $K$ refers to belief and not to knowledge, i.e. $K\varphi\rightarrow\varphi$ does not hold, then \eqref{360} may have classically false propositions as solutions. Such a proposition then says something like ``Nobody believes in me", which is false. The equations
\begin{equation*}
\begin{split}
&x\equiv\Diamond Kx\\
&x\equiv\neg\Diamond K x
\end{split}
\end{equation*}
define propositions asserting something like ``I'm believable (knowable)" and ``I'm unbelievable (unknowable)", respectively. These are further examples of non-paradoxical self-referential statements, i.e. the corresponding equations are satisfiable. 

In the last section, we present a stronger logic where epistemic operator $K$ becomes a total truth predicate. The epistemic self-referential propositions discussed here then become statements about classically truth or falsity.

\section{Relational semantics for logics extending $L5$}

In this section, we show that the S5-style modal logics of our hierarchy, i.e. those containing the modal axioms (A4) and (A5), are complete w.r.t. a relational semantics based on intuitionistic general frames. We are unable to find any kind of possible worlds semantics for weaker logics of our hierarchy. Interestingly, the presented semantic framework also describes the intuitionistic epistemic logics $IEL^-$ and $IEL$ presented in \cite{artpro}, as we shall see in the next section.\\ 
As in the preceding section, we work here with the full epistemic language $Fm$. However, dropping the epistemic ingredients from Definition \ref{790} below (more specifically: function $E$) results in (much simpler) frames for non-epistemic logic $L5$ over the modal sublanguage $Fm_1\subseteq Fm$. In this way, relational semantics for $L5$, as well as corresponding soundness and completeness proofs, are implicitly contained in the following approach.

\begin{definition}\label{790}
An $EL5^-$-frame $\mathcal{F}=(W,R,P,E,w_T)$ is given by
\begin{itemize}
\item a non-empty set $W$ of worlds
\item a partial ordering $R\subseteq W\times W$, called accessibility relation, such that $W$ has an $R$-smallest element $w_\bot$ (the bottom of the frame) and every $R$-chain has an upper bound in $W$ (Zorn's Lemma then ensures that each $w\in W$ accesses an $R$-maximal element); for $w\in W$ let $R(w):=\{w'\in W\mid wRw'\}$; and let $Max(W)$ be the set of all $R$-maximal elements
\item a set $P\subseteq Pow(W)$ of upper sets (recall that $A\in Pow(W)$ is an upper set if for all $w,w'\in W$: if $w\in A$ and $wRw'$, then $w'\in A$)
\item a function $E\colon W\rightarrow Pow(P)$ such that
\begin{itemize}
\item for each $w\in W$, $E(w)\subseteq P$ is a filter on $P$, i.e. $E(w)$ is a non-empty set with the following properties: if $A\in E(w)$ and $B\in E(w)$, then $A\cap B\in E(w)$; and if $A\in E(w)$ and $A\subseteq B\in P$, then $B\in E(w)$
\item for all $w,w'\in W$: $wRw'$ implies $E(w)\subseteq E (w')$; i.e. $E$ is a monotonic function on $W$
\end{itemize}
\item a designated $R$-maximal element $w_T\in W$.
\end{itemize}
Furthermore, we require that $P$ is closed under the following conditions:\\
(a) $\varnothing, W\in P$\\
(b) If $A,B\in P$, then the following sets are elements of $P$: 
\begin{equation*}
\begin{split}
&A\cap B\\
&A\cup B\\
&A\supset B := \{w\in W\mid\text{ for all }w'\in R(w), w'\in A\text{ implies }w'\in B\}\\
&KA := \{w\in W\mid A\in E(w)\}.
\end{split}
\end{equation*}
\end{definition}

Intuitively, $P$ is viewed as the set of all propositions. For each $w\in W$, the elements of filter $E(w)\subseteq P$ are the propositions believed at world $w$. For $A\in P$, the set $KA=\{w\in W\mid A\in E(w)\}$ is the proposition saying ``$A$ is believed (known)". Note that $KA$ is an upper set because $E$ is a monotonic function on $W$. Also note that because of $W\in P$ it holds that $W\in E(w)$, for any $w\in W$.

\begin{definition}\label{800}
Let $\mathcal{F}=(W,R,P,E,w_T)$ be an $EL5^-$-frame. $\mathcal{F}$ is
\begin{itemize}
\item an $E4L5^-$-frame if $A\in E(w)$ implies $\{w'\in W\mid A\in E(w')\}\in E(w)$, for any $w\in W$ and $A\in P$.
\item an $E5L5^-$-frame if $\mathcal{F}$ is an $E4L5^-$-frame and for every $w\in W$ and every $A\in P$: if $A\notin E(w')$ for all $w'\in R(w)$, then the proposition $\{w''\in W\mid A\in E(w'')\}\supset \varnothing$ belongs to $E(w)$.
\item an $E6L5^-$-frame if for all $w\in W$, $E(w)=E(w_\bot)$, where $w_\bot$ is the bottom world. That is, $E\colon W\rightarrow Pow(P)$ is a constant function and there is only one (global) set of believed propositions $E(w_\bot)$ which we simply denote by $E := E(w_\bot)$.
\item an $EkL5$-frame, for $k\in\{4,5,6\}$, if $\mathcal{F}$ is an $EkL5^-$-frame and for every $w\in W$, each element of $E(w)$ contains all those maximal worlds which are accessible from $w$: 
\begin{equation*}
Max(W)\cap R(w)\subseteq A,\text{ for each }A\in E(w).\footnote{It follows in particular that $\varnothing\notin E(w)$, i.e. $E(w)$ is a proper filter on $P$.}
\end{equation*}
\end{itemize}
\end{definition}

The condition of an $E4L5^-$-frame says that whenever a proposition $A$ is believed at $w$, then the proposition ``$A$ is believed" is believed at $w$. The $E5L5^-$-condition says that if proposition $A$ is unbelievable from the point of view of world $w$, then the proposition ``$A$ is unbelievable" is believed at $w$. In an $E6L5^-$-frame, a proposition $A$ is believed at some world iff $A$ is believed at all worlds iff $A\in E$. Consequently, for any $A\in P$, the proposition $KA=\{w\in W\mid A\in E(w)\}$, ``$A$ is believed", is either given by the whole set $W$ or by the empty set. Finally, the condition of an $EkL5$-frame says that if a proposition $A$ is believed at world $w$, then $A$ is true at all maximal worlds accessible from $w$. This is equivalent to the following: If a proposition $A$ is believed at world $w$, then for any world $w'$ accessible from $w$, $A$ cannot be false at $w'$. This is a semantic counterpart of intuitionistic reflection, i.e. axiom scheme (IntRe). Under this condition, belief becomes knowledge. \\

It is clear by the definition that every $E5L5^-$-frame is an $E4L5^-$-frame. Furthermore:

\begin{lemma}\label{810}
Every $E6L5^-$-frame is an $E5L5^-$-frame.
\end{lemma}

\begin{proof}
Suppose we are given an $E6L5^-$-frame. Then $E=E(w)=E(w_\bot)$, for all $w\in W$. Let $A\in E$, for some $A\in P$. Then, $\{w\in W\mid A\in E(w)\}=W\in E$. So the condition of an $E4L5^-$-frame is satisfied. Now, assume $A\notin E$. That is, $A\notin E(w)$ for all $w\in W$. Hence, $\{w\in W\mid A\in E(w)\}=\varnothing$ and $(\varnothing\supset\varnothing) = W\in E$. Thus, the condition of an $E5L5^-$-frame holds, too. 
\end{proof}

An assignment (or valuation) in a given frame $\mathcal{F}=(W,R,P,E,w_T)$ is a function $g\colon V\rightarrow P$. Given a frame $\mathcal{F}$ and an assignment $g$ in $\mathcal{F}$, we call the tuple $(\mathcal{F},g)$ a relational model based on frame $\mathcal{F}$. Given a relational model $\mathcal{K}=(\mathcal{F},g)$, the relation of satisfaction $w\vDash\varphi$, read: ``$\varphi$ is true at $w$", between worlds and formulas is defined by induction on the complexity of formulas, simultaneously for all worlds of the underlying frame $\mathcal{F}$:\\

\noindent $w\nvDash \bot$\\
$w\vDash x :\Leftrightarrow w\in g(x)$\\
$w\vDash \varphi\vee\psi :\Leftrightarrow w\vDash\varphi$ or $w\vDash\psi$\\
$w\vDash \varphi\wedge\psi :\Leftrightarrow w\vDash\varphi$ and $w\vDash\psi$\\
$w\vDash \varphi\rightarrow\psi :\Leftrightarrow$ for all $w'\in R(w)$, $w'\vDash\varphi$ implies $w'\vDash\psi$\\
$w\vDash \square\varphi :\Leftrightarrow w_\bot\vDash\varphi$\\
$w\vDash K\varphi :\Leftrightarrow \varphi^*\in E(w)$, where $\varphi^* := \{w'\in W\mid w'\vDash\varphi\}$.\\

We write $(\mathcal{F},w)\vDash\varphi$ instead of $w\vDash\varphi$ when we wish to emphasize the ambient model $\mathcal{F}$. Notice that ``$\varphi$ is false at $w$" means that $w\vDash\neg\varphi$, i.e. $\varphi$ is not true at all accessible worlds.\\

We observe that the meaning of the logical connectives is defined as in Kripke semantics of IPC while the necessity operator behaves as in Kripke semantics of modal logic S5: $(w,g)\vDash\square\varphi\Leftrightarrow$ for all $w'\in W$, $(w',g)\vDash\varphi$. The latter follows from the fact that we are dealing with rooted frames in which the usual monotonicity condition of intuitionistic frames holds: formulas true at some world remain true at accessible worlds (see the next Remark).

\begin{remark}\label{915}
Let $\mathcal{K}=(\mathcal{F},g)$ be a relational model with a set $P$ of propositions. We extend the assignment $g\colon V\rightarrow P$ to a function $g\colon Fm\rightarrow P$ defining recursively $g(\bot) :=\varnothing$, $g(\varphi\vee\psi) :=g(\varphi)\cup g(\psi)$, $g(\varphi\wedge\psi) :=g(\varphi)\cap g(\psi)$, $g(\varphi\rightarrow\psi) :=g(\varphi)\supset g(\psi)$, and

\begin{equation*}
g(K\varphi) := \{w\mid g(\varphi)\in E(w)\},
\end{equation*}

\begin{equation*}
\begin{split}
g(\square\varphi) :=
\begin{cases}
& W, \text{ if } w_\bot\in g(\varphi)\\
& \varnothing, \text{ else.}
\end{cases}
\end{split}
\end{equation*}

By closure properties of $P$, it follows inductively that $g$ is well-defined, i.e. $g(\varphi)$ is an element of $P$, for any $\varphi\in Fm$. Also by induction on the complexity of formulas, simultaneously for all worlds $w\in W$, one shows that for all $w\in W$ and all $\varphi\in Fm$, $w\vDash\varphi\Leftrightarrow w\in g(\varphi)$. That is, $g(\varphi)=\{w\in W\mid w\vDash\varphi\}=\varphi^*$, for any $\varphi\in Fm$. In particular, each $\varphi^*=g(\varphi)$ is a proposition, i.e. an element of $P$. Since all propositions are upper sets, the usual monotonicity condition of intuitionistic models follows: if $w\vDash\varphi$ and $wRw'$, then $w'\vDash\varphi$.
\end{remark}

\begin{definition}\label{920}
Let $\mathcal{K}=(\mathcal{F},g)$ be a relational model with designated maximal world $w_T$, and let $\varphi\in Fm$. We say that $\mathcal{K}$ is a model of $\varphi$, or $\varphi$ is (classically) true in $\mathcal{K}$, notation: $\mathcal{K}\vDash\varphi$, if
\begin{equation*}
(\mathcal{K},w_T)\vDash\varphi,
\end{equation*}
i.e. if $\varphi$ is true at $w_T$. This notion extends in the usual way to sets of formulas. Let $\mathcal{L}$ be the logic $EL5^-$, $EL5$, $EkL5^-$, or $EkL5$, for $k\in\{4,5,6\}$. We denote by $\mathit{Mod}^r_\mathcal{L}(\Phi)$ the class of all relational models of $\Phi$ which are based on $\mathcal{L}$-frames, and we consider the following relation of logical consequence: 
\begin{equation*}
\Phi\Vdash^r_\mathcal{L}\psi :\Leftrightarrow \mathit{Mod}^r_\mathcal{L}(\Phi)\subseteq \mathit{Mod}^r_\mathcal{L}(\{\psi\}),
\end{equation*}
where $\Phi\cup\{\psi\}\subseteq Fm$.
\end{definition} 

\begin{theorem}[Soundness]\label{940}
Let $\mathcal{L}$ be the logic $EL5^-$, $EL5$, $EkL5^-$ or $EkL5$, for $k\in \{4,5,6\}$. Then for any set of formulas $\Phi\cup\{\varphi\}\subseteq Fm$, 
\begin{equation*}
\Phi\vdash_\mathcal{L}\varphi \Rightarrow \Phi\Vdash^r_\mathcal{L}\varphi.
\end{equation*}
\end{theorem}

\begin{proof}
First, we consider logic $\mathcal{L}=E5L5$. Let $\mathcal{K}$ be any model based on an $E5L5$-frame. It suffices to show: $\mathcal{K}\vDash\square\varphi$, for all axioms $\varphi$ of logic $E5L5$ (i.e., all axioms, along with application of rule AN, are sound); and $\mathcal{K}\vDash\psi\vee\neg\psi$, for all formulas $\psi$ (i.e., \textit{tertium non datur} is sound). The latter follows immediately from the fact that truth in a model is defined as satisfaction at a maximal world. For the former, we have to show that $w_\bot\vDash\varphi$, for each axiom $\varphi$, where $w_\bot$ is the bottom world. This is clear in case of theorems of IPC and their substitution-instances (the frame is also a frame for IPC). Also the cases of (A1) and (A2) follow readily.\\
(A3): It is enough to show that $w_\bot\vDash\varphi\rightarrow\psi$ implies $w_\bot\vDash \square\varphi\rightarrow\square\psi$. This follows easily from the definition of satisfaction.\\
(A4): It is enough to show that $w_\bot\vDash\varphi$ implies $w_\bot\vDash\square\varphi$. Again, this is clear by the definition of satisfaction.\\
(A5): Truth of $\neg\square\varphi$ at some world implies truth of $\neg\square\varphi$ at all worlds implies truth of $\square\neg\square\varphi$ at all worlds.\\
(KBel): $K(\varphi\rightarrow\psi)\rightarrow (K\varphi\rightarrow K\psi)$. Let $w\in W$. It is enough to show the following: 
\begin{equation*}
\text{ If }(\varphi\rightarrow\psi)^*\in E(w)\text{ and }\varphi^*\in E(w'),\text{for any }w'\in R(w),\text{ then }\psi^*\in E(w').
\end{equation*}
Suppose the premises hold true and let $w'\in R(w)$. Then $E(w)\subseteq E(w')$. Thus, $(\varphi\rightarrow\psi)^*\in E(w')$. Since $E(w')$ is a filter, it follows that $A:=(\varphi\rightarrow\psi)^*\cap\varphi^*\in E(w')$ and $A \subseteq\psi^*\in E(w')$.\\
(CoRe): $\square\varphi\rightarrow \square K\varphi$. It is enough to show that $w_\bot\vDash\varphi$ implies $w_\bot\vDash K\varphi$. Suppose $w_\bot\vDash\varphi$. Then $\varphi^*=W$. Moreover, $W\in E(w)$, for every $w\in W$. In particular, $\varphi^*\in E(w_\bot)$. Thus, $w_\bot\vDash K\varphi$.\\
(IntRe): $K\varphi\rightarrow\neg\neg\varphi$. Suppose $w\vDash K\varphi$. Then $\varphi^*\in E(w)$. Since we are dealing with an $E5L5$-frame, $\varphi^*$ contains all maximal worlds accessible from $w$, i.e. $w'\vDash\varphi$, for all $w'\in Max(W)\cap R(w)$. Then for all $w''\in W$ accessible from $w$, we have $w''\nvDash\neg\varphi$. Hence, $w\vDash\neg\neg\varphi$. \\
(E4): Let $w\vDash K\varphi$. Then $\varphi^*\in E(w)$. By the property of an $E4L5$-frame, $\{w'\in W\mid\varphi^*\in E(w')\}=\{w'\in W\mid w'\vDash K\varphi\}=(K\varphi)^*\in E(w)$. Thus, $w\vDash KK\varphi$.\\
(E5): Let $w\vDash \neg K\varphi$. Then for all $w'\in R(w)$, $\varphi^*\notin E(w')$. By properties of an $E5L5$-frame, $A:= (\{w''\in W\mid \varphi^*\in E(w'')\}\supset\varnothing ) \in E(w)$. \\
\textit{Claim}: $A=(\neg K\varphi)^*$.\\
\textit{Proof of Claim}. We have $(\neg K\varphi)^*=\{w'''\in W\mid w'''\vDash\neg K\varphi\}=\{w'''\in W\mid\varphi^*\notin E(w'')$ whenever $w''\in R(w''')\}$. Now, by the definition of a proposition of the form $B_1\supset B_2$, one easily checks that the Claim is true.\\
So by the Claim, $(\neg K\varphi)^*\in E(w)$. That is, $w\vDash K\neg K\varphi$.\\
Finally, we consider the case of logic $\mathcal{L}=E6L5$. It remains to show that the axiom schemes (PNB) and (NNB) are valid in the class of all models based on $E6L5$-frames. But this is clear since in any $E6L5$-frame there is only one global set $E$ of known propositions: a proposition $A$ is known at some world of the frame iff $A$ is known at all worlds of the frame.
\end{proof} 

Towards the completeness theorem, we show that for any algebraic model of some of our S5-style logics there is a relational model that satisfies precisely the same set of formulas. Completeness w.r.t. relational semantics then will follow from completeness w.r.t. algebraic semantics. 

\begin{theorem}\label{960}
Let $\mathcal{L}\in\{EL5^-, EL5, EkL5^-, EkL5\mid k\in\{4,5,6\}\}$ and let $(\mathcal{M},\gamma)$ be an algebraic $\mathcal{L}$-interpretation. Then there is a relational model $\mathcal{K}=(\mathcal{F},g)$, based on a $\mathcal{L}$-frame $\mathcal{F}$, such that for all $\varphi\in Fm$:
\begin{equation*}
(\mathcal{M},\gamma)\vDash\varphi\Leftrightarrow\mathcal{K}\vDash\varphi.
\end{equation*}
\end{theorem}

\begin{proof}
We prove the assertion in detail for the case $\mathcal{L}=E5L5$. The remaining cases then follow straightforwardly. Suppose we are given a $\mathcal{L}$-interpretation $(\mathcal{M},\gamma)$ with ultrafilter $\mathit{TRUE}\subseteq M$ of true propositions and filter $\mathit{BEL}\subseteq\mathit{TRUE}$ of known propositions. Let $W$ be the set of all prime filters of the Heyting algebra reduct of $\mathcal{M}$. For $w,w'\in W$, we define $wRw':\Leftrightarrow w\subseteq w'$. Then $W$ is partially ordered by $R$, $w_T:=\mathit{TRUE}$ is a maximal element and $w_\bot := \{f_\top\}$ is the smallest element, i.e. the `bottom world'. In fact, the Disjunction Property of a model ensures that the smallest filter $\{f_\top\}$ is prime. Recall that the union of a non-empty chain of prime filters is again a prime filter. Thus, every $R$-chain in $W$ has an upper bound in $W$. For $w\in W$ put
\begin{equation*}
\mathit{BEL}(w):=\{m\in M\mid f_K(m)\in w\}.
\end{equation*}
Obviously, $\mathit{BEL}=\mathit{BEL}(\mathit{TRUE})$.
For $m\in M$, we define
\begin{equation*}
m^+:=\{w\in W\mid m\in w\}.
\end{equation*}
The set of propositions of the desired frame is
\begin{equation*}
P:=\{m^+\subseteq W\mid m\in M\},
\end{equation*}
and the set of propositions known at world $w$ is
\begin{equation*}
E(w):=\{m^+\subseteq W\mid m\in\mathit{BEL(w)}\}.
\end{equation*}
In the following, we show that $\mathcal{F}=(W,R,P,E,w_T)$ is an $EL5^-$-frame. We have to check that $\mathcal{F}$ satisfies all conditions of Definition \ref{790}. It is clear that the elements $m^+\in P$ are upper sets under inclusion, i.e. under $R$. Suppose $wRw'$, i.e. $w\subseteq w'$. Then clearly $\mathit{BEL}(w)\subseteq\mathit{BEL}(w')$ and thus $E(w)\subseteq E(w')$. The mapping $m\mapsto m^+$ defines a one-to-one correspondence between the propositions $m\in M$ of the algebraic model and the propositions $m^+\in P$ of the frame.\footnote{Surjectivity is clear. Towards injectivity suppose $m^+=m'^+$, i.e. $m$ and $m'$ are contained in exactly the same prime filters. Item (d) of Lemma \ref{200} then implies, $f_\rightarrow(m,m')=f_\top$ and $f_\rightarrow(m',m)=f_\top$ (recall that $\{f_\top\}$ is the smallest prime filter). But this means that $m\le m'$ and $m'\le m$, i.e. $m=m'$.} By properties of prime filters (see also item (d) of Lemma \ref{200}) and by the definitions, it follows that for all $m,m'\in M$:
\begin{equation*}
\begin{split}
&m^+\cap m'^+=f_\wedge(m,m')^+\\
&m^+\cup m'^+=f_\vee(m,m')^+\\ 
&m^+\supset m'^+=f_\rightarrow(m,m')^+\\
&K(m^+)=\{w\in W\mid m^+\in E(w)\}=f_K(m)^+
\end{split}
\end{equation*}
Of course, $P$ also contains $\varnothing=(f_\bot)^+$ and $W=(f_\top)^+$ and thus satisfies the closure conditions established in Definition \ref{790}. Furthermore, it follows that $(P,\cup,\cap,\supset,\varnothing, W)$ forms, in the obvious way, a Heyting algebra with least and greatest elements $\varnothing$, $W$, respectively. Although not necessary for this proof, we may consider the following additional operations on that Heyting algebra:
\begin{equation*}
f_K^P(m^+) := K(m^+)=\{w\in W\mid m^+\in E(w)\}=f_K(m)^+\\
\end{equation*}
\begin{equation*}
f_\square^P(m^+):=
\begin{cases}
\begin{split}
&W=(f_\top)^+=f_\square(m)^+\text{ if }m=f_\top\\
&\varnothing=(f_\bot)^+=f_\square(m)^+\text{ if }m\neq f_\top
\end{split}
\end{cases}
\end{equation*}
for all $m\in M$, and observe that this results in a structure that is isomorphic to the original $EL5^-$-model. In fact, one easily recognizes that the map $m\mapsto m^+$ is an isomorphism between Heyting algebras. Suppose $m=f_\top$. Since we are dealing with an $EL5^-$-model, we have $f_\square(m)=f_\top$ and thus, by truth conditions (vii) and (ii) of an algebraic model, $f_K(m)=f_\top\in w$, as $w$ is a filter. By definition of $\mathit{BEL}(w)$, $m=f_\top\in \mathit{BEL}(w)$. We have shown that for every $w\in W$, $\mathit{BEL}(w)$ contains the top element $f_\top$ of the underlying Heyting lattice, and $E(w)\neq\varnothing$. Now, in the same way as in the proof of Lemma \ref{212}, with $\mathit{TRUE}$ replaced by $w$, one shows that the sets $\mathit{BEL}(w)$ are filters on $M$. Then it follows that the sets $E(w)$ are filters on $P$. Thus, $\mathcal{F}=(W,R,P,E,w_T)$ is an $EL5^-$-frame. Let us show that $\mathcal{F}$ is an $E5L5^-$-frame according to Definition \ref{800}. Suppose $w\in W$ and $m^+\in E(w)$. Then $m\in \mathit{BEL}(w)$ and $f_K(m)\in w$. Since $\mathcal{M}$ is an $E4L5$-model, we have $f_K(m)\le f_K(f_K(m))$ and thus $f_K(f_K(m))\in w$ ($w$ is a filter). It follows that $f_K(m)^+=K(m^+)=\{w\in W\mid m^+\in E(w)\}\in E(w)$. Now suppose $w\in W$ and for all $w'\in R(w)$, $m^+\notin E(w')$. Then $f_K(m)\notin w'$ for all prime filters $w'$ extending prime filter $w$. That is, $f_\neg(f_K(m))\in w$ (see Lemma \ref{200}). Using the fact that $\mathcal{M}$ is an $E5L5^-$-model, we conclude $f_K(f_\neg(f_K(m)))\in w$. By definition of the sets $\mathit{BEL}(w)$ and $E(w)$, we get $f_\neg(f_K(m))\in\mathit{BEL}(w)$ and thus 
\begin{equation*}
\begin{split}
f_\neg(f_K(m))^+=f_\rightarrow(f_K(m), f_\bot)^+ &=f_K(m)^+\supset (f_\bot)^+\\
&=\{w''\in W\mid m^+\in E(w'')\}\supset\varnothing \in E(w).
\end{split}
\end{equation*} 
Hence, $\mathcal{F}$ is an $E5L5^-$-frame. Moreover, since $\mathcal{M}$ is a model of knowledge, we have $f_K(m)\le f_\neg(f_\neg(m))$ for all $m\in M$. So if $m\in\mathit{BEL}(w)$, then $f_K(m)\in w$ and thus $f_\neg(f_\neg(m))\in w$. That is, $m\in\mathit{BEL}(w)$ implies that $m$ belongs to all maximal worlds, i.e. ultrafilters, accessible from $w$:
\begin{equation*}
m\in\mathit{BEL}(w)\Rightarrow m\in w'\text{ for all }w'\in Max(W)\cap R(w).
\end{equation*}
Hence,
\begin{equation*}
m^+\in E(w)\Rightarrow Max(W)\cap R(w)\subseteq m^+.
\end{equation*}
In particular, $\varnothing\notin E(w)$ and $E(w)$ is a proper filter on $P$, for every $w\in W$. Thus, $\mathcal{F}=(W,R,P,E,w_T)$ is an $E5L5$-frame in the sense of Definition \ref{800}. Now, we define the following assignment $g\colon V\rightarrow P$ in $\mathcal{F}$:
\begin{equation*}
g(x) :=\gamma(x)^+, 
\end{equation*}
for each $x\in V$. Using induction, Remark \ref{915} and the previous results, one shows that 
\begin{equation*}
g(\varphi)=\gamma(\varphi)^+,
\end{equation*}
for all $\varphi\in Fm$. For instance, using the induction hypothesis, we have $\gamma(K\psi)^+=f_K(\gamma(\psi))^+=\{w\in W\mid \gamma(\psi)^+\in E(w)\}=\{w\in W\mid g(\varphi)\in E(w)\}=g(K\psi)$. We consider the relational model $\mathcal{K}=(\mathcal{F},g)$. By Remark \ref{915}, $g(\varphi)=\varphi^*=\gamma(\varphi)^+$. So we have for all $w\in W$ and all $\varphi\in Fm$: 
\begin{equation*}
w\vDash\varphi\Leftrightarrow w\in\varphi^*\Leftrightarrow w\in\gamma(\varphi)^+\Leftrightarrow\gamma(\varphi)\in w.
\end{equation*}

In particular, for the designated maximal world $w_T=\mathit{TRUE}$: 
\begin{equation*}
\mathcal{K}\vDash\varphi\Leftrightarrow w_T\vDash\varphi\Leftrightarrow\gamma(\varphi)\in w_T=\mathit{TRUE}\Leftrightarrow(\mathcal{M},\gamma)\vDash\varphi.
\end{equation*}
We have proved the assertion of the Theorem for the case $\mathcal{L}=E5L5$ and, implicitly, also for the cases $\mathcal{L}\in\{E5L5^-, E4L5^-, E4L5\}$. Finally, let us consider the cases $\mathcal{L}\in\{E6L5^-, E6L5\}$. We suppose that $\mathcal{M}$ is an $E6L5$-model. Applying the above construction, it suffices to show that the resulting function $E\colon W\rightarrow Pow(P)$ is constant, i.e. $E(w)=E(w_\bot)$, for all $w\in W$. By properties of an $E6L5^-$-model, for every $m\in M$, there are exactly two possibilities: either $f_K(m)=f_\top$ or $f_K(m)=f_\bot$, see the remark following Definition \ref{210}. It follows that for any $w\in W$: $m\in\mathit{BEL}(w)\Leftrightarrow f_K(m)\in w\Leftrightarrow f_K(m)=f_\top\Leftrightarrow m\in\mathit{BEL}$. That is, $\mathit{BEL}(w)=\mathit{BEL}$ for all $w\in W$ and function $E$ is constant.
\end{proof}

\begin{corollary}[Completeness w.r.t. relational semantics]\label{1000}
Let $\mathcal{L}$ be $EL5^-$, $EL5$, $EkL5^-$ or $EkL5$, for $k\in\{4,5,6\}$.
Then for any $\Psi\cup\{\chi\}\subseteq Fm$, 
\begin{equation*}
\Psi\Vdash^r_\mathcal{L}\chi \Rightarrow \Psi\vdash_\mathcal{L}\chi.
\end{equation*} 
\end{corollary}

\begin{proof}
Suppose $\Psi\nvdash_\mathcal{L}\chi$. By standard arguments, the set $\Psi\cup\{\neg\chi\}$ is consistent in classical logic $\mathcal{L}$. By algebraic completeness, we know that there is some algebraic $\mathcal{L}$-interpretation $(\mathcal{M},\gamma)$ satisfying that set. By Theorem \ref{960}, there is a relational model $\mathcal{K}$ based on a $\mathcal{L}$-frame such that $\mathcal{K}\vDash\Psi\cup\{\neg\chi\}$. This shows $\Psi\nVdash^r_\mathcal{L}\chi$.
\end{proof}

Once more, we point out that the above constructions specialize straightforwardly to soundness and completeness proofs for logic $L5$ w.r.t. relational semantics (considering the modal sublanguage $Fm_1$). Relational semantics for $L5$ is defined exactly as above (Definition \ref{790} and the following definition of the satisfaction relation), though, without the epistemic components, i.e. without function $E$ and operator $K$.

\section{A new relational semantics for $IEL^-$ and $IEL$}

In this section, we work with the pure epistemic sublanguage $Fm_e :=\{\varphi\in Fm\mid$ symbol $\square$ does not occur in $\varphi\}$. As already mentioned above, the intuitionistic epistemic logics $IEL^-$ and $IEL$, introduced by Artemov and Protopopescu \cite{artpro}, can be axiomatized in language $Fm_e$ by the axioms (INT), distribution of belief (KBel) $K(\varphi\rightarrow\psi)\rightarrow (K\varphi\rightarrow K\psi)$, intuitionistic co-reflection $\varphi\rightarrow K\varphi$, and -- only in case of $IEL$ -- intuitionistic reflection (IntRe) $K\varphi\rightarrow\neg\neg\varphi$. The only reference rule is Modus Ponens MP. In \cite{artpro} it is shown that $IEL^-$ and $IEL$ are sound and complete w.r.t. possible worlds semantics based on intuitionistic Kripke models. In this section, we show that these logics are sound and complete w.r.t. relational semantics of the kind presented in the preceding section. More precisely, $IEL^-$ and $IEL$ are complete w.r.t. classes of special $EL5^-$- and $EL5$-frames, respectively, which are now interpreted from the intuitionistic instead of the classical point of view. Consequently, the systems of Intuitionistic Epistemic Logic introduced in \cite{artpro} and the modal logics $EL5^-$ and $EL5$ (and their extensions) can be described within the same semantic framework. 

\begin{definition}\label{1100}
An $IEL^-$-frame (an $IEL$-frame) $\mathcal{F}=(W,R,P,E)$ is defined in exactly the same way as an $EL5^-$-frame (an $EL5$-frame), respectively (see Definition \ref{790}), but without a designated maximal world and with the following additional condition of intuitionistic co-reflection:
\begin{equation*} 
\text{(IntCo) For every }w\in W \text{ and for all propositions }A\in P: w\in A \Rightarrow A\in E(w).
\end{equation*}
\end{definition}

Intuitively, (IntCo) says that whenever a proposition $A$ is true at some world $w$, then $A$ is believed/known at $w$. This is a rather strong condition which, in particular, implies positive and negative introspection, as the next result shows.

\begin{lemma}\label{1120}
Every $IEL^-$-frame ($IEL$-frame) is an $E5L5^-$-frame ($E5L5$-frame), respectively. That is, the axioms of positive and negative introspection, (E4) and (E5), are valid.
\end{lemma}

\begin{proof}
Let $\mathcal{F}=(W,R,P,E)$ be an $IEL^-$-frame. It remains to show that the conditions of an $E5L5^-$-frame of Definition \ref{800} are satisfied. Let $w\in W$ and suppose $A\in E(w)$. Then $w\in KA=\{w'\in W\mid A \in E(w')\}$, i.e. proposition $KA$ is true at $w$. By condition (IntCo), $KA\in E(w)$. Thus, the condition of an $E4L5^-$-frame is satisfied. Now, suppose $A\notin E(w')$ for all $w'\in R(w)$. Then $w\in \neg KA=(\{w''\in W\mid A \in E(w'')\}\supset\varnothing)$. By condition (IntCo), $\neg KA\in E(w)$. Thus, the condition of an $E5L5^-$frame is satisfied.
\end{proof}

As before, an assignment (or valuation) in an $IEL^-$-frame $\mathcal{F}=(W,R,P,E)$ is a function $g\colon V\rightarrow P$. A relational $IEL^-$-model ($IEL$-model) is a tuple $\mathcal{K}=(\mathcal{F},g)$ where $\mathcal{F}$ is an $IEL^-$-frame ($IEL$-frame), respectively, and $g$ is a corresponding assignment. Also the relation of satisfaction $w\vDash\varphi$ between worlds $w\in W$ and formulas $\varphi\in Fm_e$ is defined as before, though without the clause regarding the $\square$-operator.\\

Of course, the concept of intuitionistic truth in a frame-based model should differ from the concept of classical truth in such a model. Instead of a designated maximal world, we now define truth in a frame relative to the bottom world.

\begin{definition}\label{1200}
Let $\mathcal{F}$ be an $IEL^-$-frame with bottom world $w_\bot$ and let $g\colon V\rightarrow P$ be an assinment. The notion of ``formula $\varphi\in Fm_e$ is true in model $(\mathcal{F},g)$" is defined as follows:
\begin{equation*}
(\mathcal{F},g)\vDash\varphi :\Leftrightarrow w_\bot\vDash\varphi.
\end{equation*}
We say that $\mathcal{K}=(\mathcal{F},g)$ is a (relational) $IEL^-$-model of $\varphi$ if $\varphi$ is true in $\mathcal{K}$. 
\end{definition}

\begin{theorem}[Soundness]\label{1200}
Every theorem of $IEL^-$ is true in all relational $IEL^-$-models, and every theorem of $IEL$ is true in all relational $IEL$-models.
\end{theorem}

\begin{proof}
We consider logic $IEL$. Theorems of IPC and their substitution-instances are true in relational models because such models are based on intuitionistic Kripke frames. Let us show that intuitionistic co-reflection $\varphi\rightarrow K\varphi$ is valid. Suppose we are given a relational model based on an $IEL$-frame and $w\vDash\varphi$, for some world $w\in W$. Then $w\in\varphi^*$ and the semantic condition (IntCo) of an $IEL$-frame yields $\varphi^*\in E(w)$, i.e. $w\vDash K\varphi$.
Validity of intuitionistic reflection (IntRe) and distribution of knowledge (KBel) is shown in exactly the same way as in the proof of Theorem \ref{940}. 
\end{proof}

Towards completeness, we follow a similar strategy as before. That is, we reduce completeness w.r.t. relational semantics to completeness w.r.t. algebraic semantics. We proved in [Theorem 5.3, \cite{lewigpl}] that $IEL^-$ and $IEL$ are sound and complete w.r.t. corresponding algebraic semantics. For convenience, we quote here the definition of that algebraic semantics from \cite{lewigpl}: 

\begin{definition}\label{1220}\cite{lewigpl}
An algebraic $IEL^-$-model is a Heyting algebra 
\begin{equation*}
\mathcal{M}=(M, \mathit{BEL}, f_\bot, f_\top, f_\vee, f_\wedge, f_\rightarrow, f_K)
\end{equation*}
with propositional universe $M$, a set $\mathit{BEL}\subseteq M$ of believed propositions and an additional unary operation $f_K$ such that for all propositions $m,m'\in M$ the following truth conditions hold:
\begin{enumerate}
\item $f_\top\in\mathit{BEL}$
\item $f_K(m)=f_\top\Leftrightarrow m\in\mathit{BEL}$
\item $m\le f_K(m)$
\item $f_K(f_\rightarrow(m,m'))\le f_\rightarrow(f_K(m),f_K(m'))$
\item $f_\vee(m,m')=f_\top$ $\Rightarrow$ ($m=f_\top$ or $m'=f_\top$)
\end{enumerate}
If additionally $f_K(m)\le f_\neg(f_\neg(m))$ holds for all $m\in M$, then we call $\mathcal{M}$ an $IEL$-model and $\mathit{BEL}$ is the set of known propositions.
\end{definition}

The notion of an assignment $\gamma\colon V\rightarrow M$ in an $IEL^-$-model is given as usual. We refer to a tuple $(\mathcal{M},\gamma)$ as an $IEL^-$-interpretation ($IEL$-interpretation) if $\mathcal{M}$ is an algebraic $IEL^-$-model ($IEL$-model), respectively, and $\gamma$ is a corresponding assignment. Satisfaction (truth) of a formula $\varphi\in Fm_e$ in an $IEL^-$-interpretation $(\mathcal{M},\gamma)$ is defined as follows:
\begin{equation*}
(\mathcal{M},\gamma)\vDash\varphi :\Leftrightarrow\gamma(\varphi)=f_\top.
\end{equation*}

We quote the soundness and completeness results (in weak form) from \cite{lewigpl}:

\begin{theorem}[\cite{lewigpl}]\label{1230}
Let $\varphi\in Fm_e$. Then $\varphi$ is a theorem of $IEL^-$ (of $IEL$) iff $\varphi$ is true in all algebraic $IEL^-$-interpretations (in all algebraic $IEL$-interpretations), respectively.
\end{theorem}

The next result is an analogue to Theorem \ref{960} above.

\begin{theorem}\label{1240}
Let $\mathcal{M}$ be an algebraic $IEL$-model and let $\gamma\in M^V$ be an assignment. There is a relational $IEL$-model $\mathcal{K}=(\mathcal{F},g)$ such that for all $\varphi\in Fm_e$:
\begin{equation*}
(\mathcal{M},\gamma)\vDash\varphi\Leftrightarrow \mathcal{K}\vDash\varphi.
\end{equation*}
\end{theorem}

\begin{proof}
Let $\mathcal{M}$ be an algebraic $IEL$-model with set $\mathit{BEL}$ of believed propositions, and let $\gamma\in M^V$ be an assignment in $\mathcal{M}$. The construction of an $IEL$-frame $\mathcal{F}$ from the given algebraic $IEL$-model works nearly in the same way as in the proof of Theorem \ref{960}, where an $EL5$-frame is constructed from a given algebraic $EL5$-model. The role of the designated `maximal world' $w_T=\mathit{TRUE}$ now is played by the `bottom world' $w_\bot=\{f_\top\}$. Also note that $\mathit{BEL}=\mathit{BEL}(w_\bot)$. The frame $\mathcal{F}=(W,R,P,E)$ then is given in exactly the same way as in the proof of Theorem \ref{960}, but without designated maximal world $w_T$. From the definition of an algebraic $IEL$-model (Definition \ref{1220}) it follows straightforwardly that the sets $\mathit{BEL}(w)$ and $E(w)$ are filters. All the remaining conditions of an $EL5$-frame are checked as in the proof of Theorem \ref{960} (we may skip the part of the proof where the additional conditions of an $E5L5^-$-frame are verified). Then, as before, we arrive at the following conclusions. For all $w\in W$ and all $\varphi\in Fm_e$: 
\begin{equation*}
w\vDash\varphi\Leftrightarrow w\in\varphi^*\Leftrightarrow w\in\gamma(\varphi)^+\Leftrightarrow\gamma(\varphi)\in w.
\end{equation*}
In particular, for the bottom world $w_\bot$, assignment $g(x) := \gamma(x)^+$ and relational $IEL$-model $\mathcal{K}=(\mathcal{F},g)$:
\begin{equation*}
\mathcal{K}\vDash\varphi\Leftrightarrow w_\bot\vDash\varphi\Leftrightarrow\gamma(\varphi)\in w_\bot=\{f_\top\}\Leftrightarrow\gamma(\varphi)= f_\top\Leftrightarrow(\mathcal{M},\gamma)\vDash\varphi.
\end{equation*}
It remains to show that the $EL5$-frame $\mathcal{F}$ satisfies the additional condition (IntCo) of an $IEL$-frame. Let $w\in W$ be a prime filter of the algebraic model and let $m^+\in P$ be a proposition such that $w\in m^+$, i.e. $m\in w$. By truth condition (iii) of an algebraic model (see Definition \ref{1220}), it follows that $f_K(m)\in w$, since $w$ is a filter. Then, by the definitions, $m\in \mathit{BEL}(w)$ and $m^+\in E(w)$. We have shown that (IntCo) holds.
\end{proof}

It is clear that the assertion of Theorem \ref{1240} remains true if we replace $IEL$ with $IEL^-$. Finally, we obtain (weak) soundness and completeness of the intuitionistic epistemic systems of \cite{artpro} w.r.t. our relational semantics.

\begin{corollary}[Completeness of $IEL^-$ and $IEL$ w.r.t. relational semantics]\label{1250}
A formula $\varphi\in Fm_e$ is a theorem of $IEL^-$ (of $IEL$) iff $\varphi$ is true in all relational $IEL^-$-models ($IEL$-models), respectively. 
\end{corollary}

\begin{proof}
Soundness is shown in Theorem \ref{1200}. Completeness follows from the Theorems \ref{1230} and \ref{1240}.
\end{proof}

\section{Final remarks}

We have further investigated a hierarchy of classical modal logics, originally presented in \cite{lewigpl}, for the reasoning about intuitionistic truth (proof), belief and knowledge. The axioms of the S5-style logics of our hierarchy are validated by an extended constructive BHK interpretation. Moreover, we proved soundness and completeness of those S5-style logics w.r.t. a relational semantics based on intuitionistic general frames. These results confirm our modal and epistemic axioms as adequate principles for the reasoning about proof, belief and knowledge. We have seen that the framework of relational semantics can also be used to describe the intuitionistic epistemic logics introduced by Artemov and Protopopescu \cite{artpro}. The precise relationship between our classical S5-style systems and the intuitionistic epistemic logics of \cite{artpro} becomes now explicit within that uniform semantic framework. The verification-based approach to intuitionistic belief and knowledge of \cite{artpro} turns out to be a special case of the justification-based view on belief and knowledge proposed in the present paper. From the epistemic point of view, the essential difference consists in the axiom of intuitionistic co-reflection $\varphi\rightarrow K\varphi$ (besides the fact that \textit{tertium non datur} holds in our modal logics but not in IEL). That axiom corresponds to the semantic condition (IntCo): if $w\in A$, then $A\in E(w)$, i.e. if proposition $A$ is true at world $w$, then $A$ is known at $w$. This condition is not satisfied in general but must be imposed as an additional semantic constraint on our frames. On the other hand, our modal version of co-reflection (CoRe), $\square\varphi\rightarrow\square K\varphi$, is validated by our frame-based semantics without any further assumptions. Indeed, axiom (CoRe) corresponds to the semantic condition 
\begin{equation*}
w_\bot\vDash\varphi\Rightarrow w_\bot\vDash K\varphi,
\end{equation*}
where $w_\bot$ is the bottom world. This is a property of all frames, warranted by the definitions (in fact, $w_\bot\vDash\varphi$ implies $\varphi^*= W\in E(w)$, for any world $w$, by definition of a frame). On the other hand, validity of intuitionistic co-reflection $\varphi\rightarrow K\varphi$ is equivalent to the following stronger condition: 
\begin{equation*}
\text{for all } w\in W, w\vDash\varphi\text{ implies }w\vDash K\varphi.
\end{equation*}
This is not a general property of our frames. It must be forced by the additional condition (IntCo). That is, original intuitionistic co-reflection $\varphi\rightarrow K\varphi$ is strictly stronger than its modal version $\square\varphi\rightarrow\square K\varphi$ (of course, both are interpreted intuitionistically in the sense that satisfaction at the bottom world $w_\bot$ is considered). \\

Recall that the axioms (E4), (E5), (PNB) and (NNB) are not validated by the proposed extended BHK semantics. Instead, they are considered as additional stronger epistemic principles. For instance, (PNB) along with (A2) and rule (AN) implies $K\varphi\equiv\square K\varphi$ which means that these formulas can be replaced by each other in every context. By Theorem \ref{130} (vii), $\square(\square K\varphi\vee\neg \square K\varphi)$ is a theorem of $EL5^-$. By replacements according to $K\varphi\equiv\square K\varphi$, we then obtain $\square(K\varphi\vee\neg K\varphi)$, i.e. $K\varphi\vee\neg K\varphi$ holds intuitionistically under the assumption of (PNB). \\

We finish our investigation by presenting an analogue to \eqref{0} (see the introductory section), i.e. to the main result of \cite{lewjlc2} where it is shown that the map $\varphi\mapsto\square\varphi$ embeds IPC into classical modal logic $L$. Actually, the proof of [Theorem 5.1, \cite{lewjlc2}] works the same way with $L5$ instead of $L$. Thus, we may formulate that result in the following way. For any set $\Phi\cup\{\varphi\}\subseteq Fm_0$:
\begin{equation}\label{1500}
\Phi\vdash_{IPC}\varphi\Leftrightarrow \square\Phi\vdash_{L5}\square\varphi.
\end{equation}

We consider in the following the logic $EL5^* := L5$+(IntCo)+(IntRe)+(KBel) which results from $L5$ by extending the modal language $Fm_1$ to the full language $Fm$ and by adding the axiom schemes $\varphi\rightarrow K\varphi$, $K\varphi\rightarrow\neg\neg\varphi$ and $K(\varphi\rightarrow\psi)\rightarrow (K\varphi\rightarrow K\psi)$. Of course, rule AN then applies also to those epistemic axioms. One recognizes that $EL5^*$ results from $EL5$ by replacing (CoRe) $\square\varphi\rightarrow \square K\varphi$ with the stronger scheme (IntCo) $\varphi\rightarrow K\varphi$. By rule AN along with distribution and Lemma \ref{80}, $\square(\square\varphi\rightarrow \square K\varphi)$ is a theorem of $EL5^*$. Thus, $EL5^*$ is strictly stronger than $EL5$. Moreover, by (IntCo) + (IntRe) and application of (TND),
\begin{equation*}
\text{(TB)     } K\varphi\leftrightarrow \varphi
\end{equation*}
is a theorem scheme of $EL5^*$. That is, the epistemic operator $K$ becomes a truth predicate of the object language, (TB) is an analogue to the Tarski biconditionals (T-scheme) of Tarski's truth theory. Considering the original soundness and completeness proofs of $EL5$, one easily checks that logic $EL5^*$ is sound and complete w.r.t. to the class of those algebraic $EL5$-models which satisfy the additional semantic condition: $m\le f_K(m)$, for all propositions $m$. In such a model, the set of known propositions coincides precisely with the set of classically true propositions, i.e. the set of facts: $\mathit{BEL}=\mathit{TRUE}$. One also verifies that $EL5^*$ is sound and complete w.r.t. relational semantics given by the class of those $EL5$-frames which satisfy the additional semantic condition (IntCo) of our $IEL$-frames introduced in the previous section (recall that condition (IntCo) corresponds to intuitionistic co-reflection, $\varphi\rightarrow K\varphi$). 

Note that $K\varphi\equiv \varphi$ is not a theorem of $EL5^*$, i.e. $K\varphi$ and $\varphi$ generally denote different propositions. In fact, in every model we have: $m\le f_K(m)\le f_\neg(f_\neg(m))$, for all propositions $m$, and there are many Heyting algebras where $f_K$ can be defined in such a way that those inequalities are strict. This means that knowledge and classical truth are equivalent (viewed as predicates, they have the same \textit{extensions}) although the formulas $K\varphi$ and $\varphi$ have, in general, different \textit{intensions, meanings}. The discussion on self-referential propositions of section 3 now can be applied to equations in $EL5^*$ involving operator $K$ as a total truth predicate for classical truth. For instance, the equation $x\equiv K x$ defines a truth-teller: any proposition that satisfies that equation says ``I am (classically) true". There are true and false truth-tellers, the equation is satisfiable in several ways. The liar proposition is claimed by $x\equiv \neg Kx$. Whenever that equation holds in some model, the proposition denoted by $x$ says ``I am not (classically) true", i.e. ``This proposition is (classically) false". Of course, no model satisfies that equation as one easily verifies by the truth conditions of a model. The liar can be stated by an equation, though the liar proposition as a semantic object does not exist. In this sense, we have a solution to the liar paradox. L\"ob's paradox can be viewed as a contingent liar and is expressed by the equation $x\equiv (K x\rightarrow\varphi)$, where $\varphi$ is any formula. If that equation is satisfied, then the proposition denoted by $x$ says ``If this proposition is true, then proposition $\varphi$ holds". The equation is satisfiable in models that satisfy $\varphi$. If $\varphi$ is false, then the equation represents an antinomy as one easily checks. These examples illustrate that the logic is able to deal with semantic antinomies by means of equations. Neither the liar paradox nor contingent liars, such as L\"ob's paradox, give rise to inconsistencies. This solution to semantic paradoxes is possible because we are working with a non-Fregean logic and strictly distinguish between formulas as syntactic objects and propositions as their meaning. Self-reference is expressed by equations on the syntactic level. Propositions (semantic objects) satisfying such equations are self-referential. Paradoxical self-referential propositions cannot exist since the corresponding self-referential equations are unsatisfiable.\footnote{This approach to self-reference was presented in \cite{str}; see also \cite{lewsl0, lewigpl0, zei} for further information.} Tarski's truth theory and many subsequent approaches do not consider such a distinction between sentences and propositions. If the language is sufficiently strong, a total truth predicate of the object language then leads to the construction of the paradoxical liar sentence and thus to the inconsistency of the underlying system.\\

\textit{In the classical extension $EL5^*$ of $IEL$, knowledge operator $K$ becomes a total truth predicate of the object language. That is, the Tarski Biconditionals (TB), formulated in the object language, are valid. Furthermore, $IEL$ corresponds to $EL5^*$ in a similar way as IPC corresponds to $L5$, i.e. the following holds:}

\begin{theorem}\label{2000}
For any set $\Phi\cup\{\varphi\}\subseteq Fm_e$ of propositional epistemic formulas:
\begin{equation*}
\Phi\vdash_{IEL}\varphi\Leftrightarrow \square\Phi\vdash_{{EL5}^*}\square\varphi.
\end{equation*}
In particular, the embedding $\varphi\mapsto\square\varphi$ of IPC into $L5$ (into $L$) extends to an embedding of $\mathit{IEL}$ into $EL5^*$. Thus, $EL5^*$ contains a copy of $\mathit{IEL}$ in the form of $\{\square\varphi\mid \varphi\in Fm_e$ is a theorem of $IEL\}$.
\end{theorem}

\begin{proof}
The left-to-right-implication of the Theorem follows straightforwardly by induction on the length of a derivation of $\varphi$ from $\Phi$ in $IEL$. Towards the right-to-left-implication, we suppose $\Phi\nvdash_{IEL}\varphi$. By completeness of $IEL$ w.r.t. our relational semantics, there exists a model $\mathcal{K}$ based on an $IEL$-frame (Definition \ref{1100}) that satisfies $\Phi$ but not $\varphi$. By definition, an $IEL$-frame (with a chosen designated maximal world) is an $EL5$-frame satisfying the additional condition (IntCo). That is, choosing a designated maximal world, we may interpret $\mathcal{K}$ as a model of $EL5^*$ (with our S5-style reading of modal operator $\square$). Since $\mathcal{K}$ viewed as an $IEL$-model satisfies $\Phi$ at the bottom world, $\mathcal{K}$ viewed as an $EL5^*$-model satisfies $\square\Phi$ at every world -- in particular, at the designated maximal world. By similar arguments, $\square\varphi$ is satisfied at no (maximal) world. Thus, $\square\Phi\nVdash_{EL5^*}\square\varphi$. Since $EL5^*$ is sound and complete w.r.t. the class of all $EL5$-frames satisfying the semantic condition (IntCo), we conclude $\square\Phi\nvdash_{EL5^*}\square\varphi$.
\end{proof}

Because of (IntCo) $\varphi\rightarrow K\varphi$, the following Generalization Rule holds in $EL5^*$: ``If $\varphi$ is a theorem of $EL5^*$, then so is $K\varphi$." Note that positive and negative introspection, $K\varphi\rightarrow K K\varphi$ and $\neg K\varphi\rightarrow K\neg K\varphi$, respectively, are particular instances of the scheme of intuitionistic co-reflection and are therefore theorems of $IEL$ as well as of $EL5^*$. It follows that $EL5^*$ contains the classical epistemic logic KT45, i.e. the epistemic logic that corresponds to S5, where the epistemic operator $K$ plays the role of the modal operator $\square$. In fact, if we add \textit{tertium non datur} (TND) to $IEL$, then we obtain KT45+(TB). That is, KT45+(TB) extends $IEL$ in a similar way as CPC extends IPC. In this sense, KT45+(TB) can be seen as the classical counterpart of $IEL$. Let us summarize the preceding discussion:\\

\textit{Logic $EL5^*$ contains classical epistemic logic KT45 and `a copy' of intuitionistic epistemic logic $IEL$ via the embedding $\varphi\mapsto\square\varphi$. Theorem \ref{2000} can be viewed as an analogue to result \eqref{1500} shown in \cite{lewjlc2}.}

\end{document}